%% file: main.tex
\RequirePackage[l2tabu,orthodox]{nag}
\documentclass
[11pt,letterpaper]
{article}

\usepackage{etex}
\usepackage{xspace,enumerate}
\usepackage[dvipsnames]{xcolor}
\usepackage[T1]{fontenc}
\usepackage[full]{textcomp}
\usepackage[american]{babel}
\usepackage{mathtools}
\usepackage{amsthm}
\newtheorem{theorem}{Theorem}[section]
\newtheorem*{theorem*}{Theorem}

\newtheorem{proposition}[theorem]{Proposition}
\newtheorem*{proposition*}{Proposition}
\newtheorem{lemma}[theorem]{Lemma}
\newtheorem*{lemma*}{Lemma}
\newtheorem{corollary}[theorem]{Corollary}
\newtheorem*{conjecture*}{Conjecture}
\newtheorem{fact}[theorem]{Fact}
\newtheorem*{fact*}{Fact}

\newtheorem*{hypothesis*}{Hypothesis}

\theoremstyle{definition}
\newtheorem{definition}[theorem]{Definition}
\newtheorem*{definition*}{Definition}

\newtheorem{algorithm}[theorem]{Algorithm}

\theoremstyle{remark}

\newtheorem*{claim*}{Claim}
\newtheorem{remark}[theorem]{Remark}
\newtheorem*{remark*}{Remark}

\newtheorem*{observation*}{Observation}
\usepackage[
letterpaper,
top=1.2in,
bottom=1.2in,
left=1in,
right=1in]{geometry}
\usepackage{newpxtext} %
\usepackage{textcomp} %
\usepackage[varg,bigdelims]{newpxmath}
\usepackage[scr=rsfso]{mathalfa}%
\usepackage{bm} %
\let\mathbb\varmathbb
\usepackage{microtype}
\usepackage[
pagebackref,
colorlinks=true,
urlcolor=blue,
linkcolor=blue,
citecolor=OliveGreen,
]{hyperref}
\usepackage[capitalise,nameinlink]{cleveref}
\crefname{lemma}{Lemma}{Lemmas}
\crefname{fact}{Fact}{Facts}
\crefname{theorem}{Theorem}{Theorems}
\crefname{corollary}{Corollary}{Corollaries}
\crefname{claim}{Claim}{Claims}
\crefname{example}{Example}{Examples}
\crefname{algorithm}{Algorithm}{Algorithms}
\crefname{problem}{Problem}{Problems}
\crefname{definition}{Definition}{Definitions}
\usepackage{paralist}
\usepackage{turnstile}
\usepackage{mdframed}
\usepackage{tikz}
\usepackage{caption}
\DeclareCaptionType{Algorithm}
\usepackage{newfloat}

\newcommand{\Authornote}[2]{}
\newcommand{\Authornotecolored}[3]{}
\newcommand{\Authorcomment}[2]{}
\newcommand{\Authorfnote}[2]{}
\newcommand{\Dnote}{\Authornote{D}}

\newcommand{\Pnote}{\Authornote{P}}

\definecolor{forestgreen(traditional)}{rgb}{0.0, 0.27, 0.13}

\usepackage{boxedminipage}

\newcommand{\Paren}[1]{\left(#1\right)}
\newcommand{\bigparen}[1]{\big(#1\big)}
\newcommand{\Bigparen}[1]{\Big(#1\Big)}

\newcommand{\Brac}[1]{\left[#1\right]}

\newcommand{\abs}[1]{\lvert#1\rvert}

\newcommand{\card}[1]{\lvert#1\rvert}

\newcommand{\set}[1]{\{#1\}}
\newcommand{\Set}[1]{\left\{#1\right\}}

\newcommand{\norm}[1]{\lVert#1\rVert}
\newcommand{\Norm}[1]{\left\lVert#1\right\rVert}

\newcommand{\iprod}[1]{\langle#1\rangle}
\newcommand{\Iprod}[1]{\left\langle#1\right\rangle}

\newcommand{\Esymb}{\mathbb{E}}
\newcommand{\Psymb}{\mathbb{P}}

\DeclareMathOperator*{\E}{\Esymb}

\DeclareMathOperator*{\ProbOp}{\Psymb}
\renewcommand{\Pr}{\ProbOp}
\newcommand{\seteq}{\mathrel{\mathop:}=}
\newcommand{\from}{\colon}
\newcommand{\mper}{\,.}
\newcommand{\mcom}{\,,}
\newcommand\bdot\bullet
\DeclareMathOperator{\Ind}{\mathbf 1}

\DeclareMathOperator{\poly}{poly}

\newcommand{\Hoelder}{H\"{o}lder\xspace}
\newcommand{\Holder}{\Hoelder}

\newcommand{\N}{\mathbb N}
\newcommand{\R}{\mathbb R}

\newcommand{\cA}{\mathcal A}
\newcommand{\cB}{\mathcal B}
\newcommand{\cC}{\mathcal C}
\newcommand{\cD}{\mathcal D}
\newcommand{\cE}{\mathcal E}

\newcommand{\cN}{\mathcal N}

\newcommand{\bbS}{\mathbb S}

\renewcommand{\leq}{\leqslant}
\renewcommand{\le}{\leqslant}
\renewcommand{\geq}{\geqslant}
\renewcommand{\ge}{\geqslant}
\let\epsilon=\varepsilon
\numberwithin{equation}{section}
\newcommand\MYcurrentlabel{xxx}
\newcommand{\MYstore}[2]{%
  \global\expandafter \def \csname MYMEMORY #1 \endcsname{#2}%
}
\newcommand{\MYload}[1]{%
  \csname MYMEMORY #1 \endcsname%
}
\newcommand{\MYnewlabel}[1]{%
  \renewcommand\MYcurrentlabel{#1}%
  \MYoldlabel{#1}%
}
\newcommand{\MYdummylabel}[1]{}
\newcommand{\torestate}[1]{%
  \let\MYoldlabel\label%
  \let\label\MYnewlabel%
  #1%
  \MYstore{\MYcurrentlabel}{#1}%
  \let\label\MYoldlabel%
}
\newcommand{\restatetheorem}[1]{%
  \let\MYoldlabel\label
  \let\label\MYdummylabel
  \begin{theorem*}[Restatement of \cref{#1}]
    \MYload{#1}
  \end{theorem*}
  \let\label\MYoldlabel
}
\newcommand{\restatelemma}[1]{%
  \let\MYoldlabel\label
  \let\label\MYdummylabel
  \begin{lemma*}[Restatement of \cref{#1}]
    \MYload{#1}
  \end{lemma*}
  \let\label\MYoldlabel
}
\newcommand{\restateprop}[1]{%
  \let\MYoldlabel\label
  \let\label\MYdummylabel
  \begin{proposition*}[Restatement of \cref{#1}]
    \MYload{#1}
  \end{proposition*}
  \let\label\MYoldlabel
}
\newcommand{\restatefact}[1]{%
  \let\MYoldlabel\label
  \let\label\MYdummylabel
  \begin{fact*}[Restatement of \prettyref{#1}]
    \MYload{#1}
  \end{fact*}
  \let\label\MYoldlabel
}
\newcommand{\restate}[1]{%
  \let\MYoldlabel\label
  \let\label\MYdummylabel
  \MYload{#1}
  \let\label\MYoldlabel
}

\newcommand{\e}{\epsilon}

\allowdisplaybreaks
\newcommand*{\Id}{\mathrm{Id}}

\newcommand*{\on}{\{\pm 1\}}
\newcommand*{\sos}{\mathrm{sos}}
\newcommand*{\U}{\mathcal{U}}

\newcommand*{\Lowner}{L\"owner\xspace}
\newcommand*{\normtv}[1]{\Norm{#1}_{\mathrm{TV}}}
\newcommand*{\normf}[1]{\Norm{#1}_{\mathrm{F}}}

\DeclareMathOperator{\pE}{\tilde{\mathbb{E}}}

\newcommand*{\dyad}[1]{#1#1{}^{\mkern-1.5mu\mathsf{T}}}

\title{
  Outlier-robust moment-estimation\\ via sum-of-squares
}

\author{
    Pravesh K. Kothari\thanks{Princeton University and Institute for Advanced Study}
      \and
    David Steurer\thanks{ETH Z\"urich. Much of this work was done while at Cornell University and the Institute for Advanced Study. Supported by a Microsoft Research Fellowship, a Alfred P. Sloan Fellowship, NSF awards, and the Simons Collaboration for Algorithms and Geometry.}
}

\begin{document}

\pagestyle{empty}

\maketitle
\thispagestyle{empty} %

\begin{abstract}

\input{content/abstract}

\end{abstract}

\clearpage

  \microtypesetup{protrusion=false}
  \tableofcontents{}
  \microtypesetup{protrusion=true}

\clearpage

\pagestyle{plain}
\setcounter{page}{1}

\input{content/introduction}

\input{content/identifiability}
\input{content/preliminaries}
\input{content/algorithm}

\input{content/subgaussianity}
\input{content/application-ica}
\input{content/mixturesofgaussians}

\input{content/lower-bounds}

  \phantomsection
  \addcontentsline{toc}{section}{References}
  \bibliographystyle{amsalpha}
  \bibliography{bib/mathreview,bib/dblp,bib/custom,bib/scholar}

\appendix

\input{content/basic-sos-proofs}

\end{document}

%% file: content/abstract.tex
We develop efficient algorithms for estimating low-degree moments of unknown distributions in the presence of adversarial outliers.
The guarantees of our algorithms improve in many cases significantly over the best previous ones, obtained in recent works of \cite{DBLP:conf/focs/DiakonikolasKK016,DBLP:conf/focs/LaiRV16,DBLP:conf/stoc/CharikarSV17}.
We also show that the guarantees of our algorithms match information-theoretic lower-bounds for the class of distributions we consider. 
These improved guarantees allow us to give improved algorithms for independent component analysis and  learning mixtures of Gaussians in the presence of outliers.

Our algorithms are based on a standard sum-of-squares relaxation of the following conceptually-simple optimization problem:
Among all distributions whose moments are bounded in the same way as for the unknown distribution, find the one that is closest in statistical distance to the empirical distribution of the adversarially-corrupted sample.

%% file: content/introduction.tex
\section{Introduction}

We consider the problem of \emph{outlier-robust parameter estimation}:
We are given independent draws $x_1,\ldots,x_n$ from an unknown distribution $D$ over $\R^d$ and the goal is to estimate parameters of the distribution, e.g., its mean or its covariance matrix.
Furthermore, we require that the estimator is \emph{outlier-robust}:
even if an $\e$-fraction of the draws are corrupted by an adversary, the estimation error should be small.

These kind of estimators have been studied extensively in statistics (under the term \emph{robust statistics}) \cite{MR0426989,maronna2006robust,huber2011robust,hampel2011robust}.
However, many robust estimators coming out of this research effort are computationally efficient only for low-dimensional distributions (because the running time is say exponential in the dimension $d$) \cite{bernholt2006robust}. 

A recent line of research developed the first robust estimators for basic parameter estimation problems (e.g., estimating the mean and covariance matrix of a Gaussian distribution) that are computationally efficient for the high-dimensional case, i.e., the running time is only polynomial in the dimension $d$ \cite{DBLP:conf/focs/DiakonikolasKK016,DBLP:conf/focs/LaiRV16,DBLP:conf/colt/Cherapanamjeri017,DBLP:conf/stoc/CharikarSV17}.

Our work continues this line of research.
We design efficient algorithms to estimate low-degree moments of distributions.
Our estimators succeed under significantly weaker assumptions about the unknown distribution $D$, even for the most basic tasks of estimating the mean and covariance matrix of $D$.
For example, in order to estimate the mean of $D$ (in an appropriate norm) our algorithms do not need to assume that $D$ is Gaussian or has a covariance matrix with small spectral norm in contrast to assumptions of previous works.
Similary, our algorithms for estimating covariance matrices work, unlike previous algorithms, for non-Gaussian distributions and distributions that are not (affine transformations of) product distributions.

Besides these qualitative differences, our algorithms also offer quantitative improvements.
In particular, for the class of distributions we consider, the guarantees of our algorithms---concretely, the asymptotic behavior of the estimation error as a function of the fraction $\e$ of corruptions---match, for the first time in this generality, information-theoretic lower bounds.

\paragraph{Outlier-robust method of moments}
Our techniques for robust estimation of mean vectors and covariance matrices extend in a natural way to higher-order moment tensors.
This fact allows us to turn many non-outlier-robust algorithms in a black-box way into outlier-robust algorithms.
The reason is that for many parameter estimation problems the best known algorithms in terms of provable guarantees are based on the \emph{method of moments}, which means that they don't require direct access to a sample from the distribution but instead only to its low-degree moments (e.g. \cite{pearson1894contributions,DBLP:conf/stoc/KalaiMV10,DBLP:conf/focs/MoitraV10}).
(Often, a key ingredient of these algorithms in the high-dimensional setting is \emph{tensor decomposition} \cite{DBLP:conf/stoc/MosselR05,DBLP:conf/nips/AnandkumarFHKL12,DBLP:journals/jmlr/AnandkumarGHKT14,DBLP:conf/innovations/HsuK13,DBLP:conf/stoc/BhaskaraCMV14,DBLP:conf/stoc/BarakKS15,DBLP:conf/approx/GeM15,DBLP:conf/focs/MaSS16}.)
If there were no outliers, we could run these kinds of algorithms on the empirical moments of the observed sample.
However, in the presence of outliers, this approach fails dramatically because even a single outlier can have a huge effect on the empirical moments.
Instead, we apply method-of-moment-based algorithms on the output of our outlier-robust moment estimators.
Following this strategy, we obtain new outlier-robust algorithms for independent component analysis (even if the underlying unknown linear transformation is ill-conditioned) and mixtures of spherical Gaussians (even if the number of components of the mixture is large and the means are not separated).

\paragraph{Estimation algorithms from identifiability proofs}
Our algorithms and their analysis follow a recent paradigm for computationally-efficient provable parameter estimation that has been developed in the context of the \emph{sum-of-squares method}.
We say that a parameter estimation problem satisfies \emph{identifiability} if it is information-theoretically possible to recover the desired parameter from the available data (disregarding computational efficiency).
The key idea of this paradigm is that a proof of identifiability can be turned into an efficient estimation algorithm if the proof is captured by a low-complexity proof system like sum-of-squares.
Many estimation algorithms based on convex relaxations, in particular sum-of-squares relaxations, can be viewed as following this paradigm (e.g., compressed sensing and matrix completion \cite{DBLP:conf/focs/CandesRTV05,DBLP:journals/focm/CandesR09,DBLP:journals/tit/Gross11,DBLP:journals/jmlr/Recht11}).
Moreover, this paradigm has been used, with a varying degree of explicitness, in order to design a number of recent algorithms based on sum-of-squares for unsupervised learning, inverse, and estimation problems like overcomplete tensor decomposition, sparse dictionary learning, tensor completion, tensor principal component analysis \cite{DBLP:conf/colt/HopkinsSS15,DBLP:conf/stoc/BarakKS15,DBLP:conf/colt/BarakM16,DBLP:conf/colt/PotechinS17,DBLP:conf/focs/MaSS16}.

In many of these settings, including ours, the proof of identifiability takes a particular form:
given a (corrupted) sample $X\subseteq \R^d$ from a distribution $D$ and a candidate parameter $\hat \theta$ that is close to true parameter $\theta$ of $D$, we want to be able to efficiently certify that the candidate parameter $\hat \theta$ is indeed close to $\theta$.
(This notion of certificate is similar to the one for NP.)
Following the above paradigm, if the certification in this identifiability proof can be accomplished by a low-degree sum-of-squares proof, we can derive an efficient estimation algorithm (that computes an estimate $\hat\theta$ just given the sample $X$).

Next we describe the form of our identifiability proofs for outlier-robust estimation.

\paragraph{Identifiability in the presence of outliers}
Suppose we are given an $\e$-corrupted sample $X=\set{x_1,\ldots,x_n}\subseteq \R^d$ of an unknown distribution $D$ and our goal is to estimate the mean $\mu\in \R^d$ of $D$.
To hope for the mean to be identifiable, we need to make some assumptions about $D$.
Otherwise, $D$ could for example have probability mass $1-\e$ on $0$ and probability mass $\e$ on $\mu/\e$.
Then, an adversary can erase all information about the mean $\mu$ in an $\e$-corrupted sample from $D$ (because only an $\e$ fraction of the draws from $D$ carry information about $\mu$).
Therefore, we will assume that $D$ belongs to some class of distributions $\cC$ (known to the algorithm).
Furthermore, we will assume that the class $\cC$ is defined by conditions on low-degree moments so that if we take a large enough (non-corrupted) sample $X^0$ from $D$, the uniform distribution over $X^0$ also belongs to the class $\cC$ with high probability.
(We describe the conditions that define $\cC$ in a paragraph below.)

With this setup in place, we can describe our robust identifiability proof:
it consists of a (multi-)set of vectors $X'=\set{x_1',\ldots,x_n'}\subseteq \R^d$ that satisfies two conditions:
\begin{enumerate}
\item $x_i'=x_i$ for all but an $\e$ fraction of the indices $i\in [n]$,
\item the uniform distribution over $X'$ is in $\cC$.
\end{enumerate}
Note that given $X$ and $X'$, we can efficiently check the above conditions, assuming that the conditions on the low-degree moments that define $\cC$ are efficiently checkable (which they will be).

Also note that the above notion of proof is complete in the following sense:
if $X$ is indeed an $\e$-corruption of a typical\footnote{The sample $X^0$ of $D$ should be typical in the sense that the empirical low-degree moments of the sample $X^0$ are close to the (population) low-degree moments of the distribution $D$.
If the sample is large enough (polynomial in the dimension), this condition is satisfied with high probability.} sample $X^0$ from a distribution $D\in \cC$, then there exists a set $X'$ that satisfies the above conditions, namely the uncorrupted sample $X^0$.

In \cref{sec:robust-identifiability}, we show that the above notion of proof is also sound:
if $X$ is indeed an $\e$-corruption of a (typical) sample $X^0$ from a distribution $D\in \cC$ and $X'$ satisfies the above conditions, then the empirical mean $\mu'\seteq\tfrac 1n \sum_{i=1}^n x'_i$ of the uniform distribution over $X'$ is close to $\mu$.
We can rephrase this soundness as the following concise mathematical statement, which we prove in \cref{sec:robust-identifiability}:
if $D$ and $D'$ are two distributions in $\cC$ that have small statistical distance, then their means are close to each other (and their higher-order moments are close to each other as well).

Furthermore, the above soundness is captured by a low-degree sum-of-squares proof (using for example the sum-of-squares version of \Holder's inequality);
this fact is the basis of our efficient algorithm for outlier-robust mean estimation (see \cref{sec:moment-estim-algor}).

\paragraph{Outlier-robustness and (certifiable) subgaussianity}
To motivate the aforementioned conditions we impose on the distributions, consider the following scenario:
we are given an $\e$-corrupted sample of a distribution $D$ over $\R$ and our goal is to estimate its variance $\sigma^2$ up to a constant factors.
To hope for the variance to be identifiable, we need to rule out for example that $D$ outputs $0$ with probability $1-\e$ and a Gaussian $N(0,\sigma^2/\e^2)$ with probability $\e$.
Otherwise, an adversary can remove all information about $\sigma$ by changing an $\e$ fraction of the draws in a sample from $D$.
If we look at the low-degree moments of this distribution $D$, we see that $(\E_{D(x)} x^k)^{1/k}\approx \sqrt{k / \e^{1-2/k}} (\E_{D(x)} x^2)^{1/2}$.
So for large enough $k$ (i.e., $k\approx \log (1/\e)$), the ratio between the $L_k$ norm of $D$ and the $L_2$  norm of $D$ exceeds that of a Gaussian by a factor of roughly $1/\sqrt \e$.
In order to rule out this example, we impose the following condition on the low-degree moments of $D$, and we show that this condition is enough to estimate the variance of a distribution $D$ over $\R$ up to constant factors given an $\e$-corrupted sample,
\begin{equation}
  \label{eq:subgaussian-scalar}
  \Bigparen{\E_{D(x)} (x-\mu_D)^k}^{1/k} \le \sqrt {C k} \cdot \Paren{\E_{D(x)} (x-\mu_D)^2 }^{1/2}\text{ for $C>0$ and even $k\in \N$ with } C k \cdot \e^{1-2/k}\ll 1
  \mper
\end{equation}
Here, $\mu_D$ is the mean of the distribution $D$.

In the high-dimensional setting, a natural idea is to impose the condition \cref{eq:subgaussian-scalar} for every direction $u\in \R^d$,
\begin{equation}
  \label{eq:subgaussian-vector}
  \Bigparen{\E_{D(x)} \iprod{x-\mu_D,u}^k}^{1/k} \le \sqrt {C k} \cdot \Paren{\E_{D(x)} \iprod{x-\mu_D,u}^2 }^{1/2}\mper
\end{equation}
In \cref{sec:robust-identifiability}, we show that this condition is indeed enough to ensure identifiability and that it is information-theoretically possible to estimate the covariance matrix $\Sigma_D$ of $D$ up to constant factors (in the L\"owner order sense) assuming again that $C k \cdot \e^{1-2/k}\ll 1$.
Unfortunately, condition \cref{eq:subgaussian-vector} is unlikely to be enough to guarantee an efficient estimation algorithm.
The reason is that \cref{eq:subgaussian-vector} might hold for the low-degree moments of some distribution $D$ but every proof of this fact requires exponential size.
(This phenomenon is related to the fact that finding a vector $u\in \R^d$ that violates \cref{eq:subgaussian-vector} is an NP-hard problem in general.)

To overcome this source of intractability, we require that inequality \cref{eq:subgaussian-vector} is not only true but also has a low-degree sum-of-squares proof.
With this additional condition, we give a polynomial-time algorithm to estimate the covariance matrix $\Sigma_D$ of $D$ up to constant factors (in the L\"owner order sense) assuming again that $C k \cdot \e^{1-2/k}\ll 1$.

\begin{definition}[Certifiable subgaussianity of low-degree moments]
  \label{def:certifiable-subgaussianity}
  A distribution\footnote{We emphasize that our notion of certifiable subgaussianity is a property only of the low-degree moments of a distribution and, in this way, significantly less restrictive then the usual notion of subgaussianity (which restricts all moments of a distribution).} $D$ over $\R^d$ with mean $\mu$ is \emph{$(k,\ell)$-certifiably subgaussian} with parameter $C>0$ if there for every positive integer $k'\le k/2$, there exists a degree-$\ell$ sum-of-squares proof\footnotemark{} of the degree-$2k'$ polynomial inequality over the unit sphere,
  \begin{equation}
    \label{eq:low-degree-subgaussian}
    \forall u\in \bbS^{d-1}.~
    \E_{D(x)} \iprod{x-\mu,u}^{2k'}
    \le \Bigparen{C\cdot k' \E_{D(x)} \iprod{x-\mu,u}^2 }^{k'}
  \end{equation}
\end{definition}

\footnotetext{
  In the special case of proving that $p(u)\ge 0$ holds for every vector $u\in\bbS^{d-1}$, a degree-$\ell$ sum-of-squares proof consists of a polynomial $q(u)$ with $\deg q\le \ell-2$ and polynomials $r_1(u),\ldots,r_t(u)$ with $\deg r_\tau \le \ell/2$ such that $p=q\cdot (\norm{u}^2-1) + r_1^2+\cdots +r_t^2$.
  See \cref{sec:preliminaries} for a more general definition.
}

In \cref{sec:subgaussianity}, we show that a wide range of distributions are certifiably subgaussian and that many operations on distributions preserve this property.
In particular, (affine transformations of) products of scalar-valued subgaussian distributions and mixtures thereof satisfy this property.
In a subsequent work, Kothari and Steinhardt \cite{KothariSteinhardt17} show that certifiable subgaussianity holds for distributions that satisfy a Poincar\'e inequality (which includes all strongly log-concave distributions). 

Since any valid polynomial inequality over the sphere has a sum-of-squares proof if we allow the degree to be large enough and the inequality has positive slack, we say that $D$ is \emph{$(k,\infty)$-certifiably subgaussian} with parameter $C>0$ if the inequalities \cref{eq:low-degree-subgaussian} hold for all $k'\le k/2$.
In most cases we consider, a sum-of-squares proof of degree $\ell=k$ is enough.
In this case, we just say that $D$ is $k$-certifiably subgaussian.
Observe that $k$-certifiable subgaussianity only restricts moments up to degree $k$ of the underlying distribution.
In this sense, it is a much weaker assumption than the usual notion of subgaussianity, which imposes Gaussian-like upper bounds on all moments.

\paragraph{Sum-of-squares and quantifier alternation}

Taking together the above discussion of robust identifiability proofs and certifiable subgaussianity, our approach to estimate the low-degree moments of a certifiable subgaussian distribution is the following:
given an $\e$-corrupted sample $X=\set{x_1,\ldots,x_n}\subseteq \R^d$ from $D$, we want to find a (multi-)set of vectors $X'=\set{x'_1,\ldots,x'_n}\subseteq \R^d$ such that the uniform distributions over $X$ and $X'$ are $\e$-close in statistical distance and the uniform distribution over $X'$ is certifiably subgaussian.

It is straightforward to formulate these conditions as a system $\cE$ of polynomial equations over the reals such that $x_1',\ldots,x'_n$ are (some of the) variables (see \cref{sec:moment-estim-algor}).
Our outlier-robust moment estimation proceeds by solving a standard sum-of-squares relaxation of this system $\cE$.
The solution to this relaxation can be viewed as a \emph{pseudo-distribution} $D'$ over vectors $x'_1,\ldots,x'_n$ that satisfies the system of equations $\cE$.
(This pseudo-distribution behaves in many ways like a classical probability distribution over vectors $x'1_,\ldots,x'_n$ that satisfy the equations $\cE$.
See \cref{sec:preliminaries} for a definition.)
Our moment estimation algorithm simply outputs the expected empirical moments $\tilde\E_{D'(x'_1,\ldots,x'_n)}\tfrac 1n \sum_{i=1}^n (x'_i)^{\otimes r}$ with respect to the pseudo-distribution $D'$.

We remark that previous work on computationally-efficiently outlier-robust estimation also used convex optimization techniques albeit in different ways.
For example, Diakonikolas et al. solve an implicitly-defined convex optimization problem using a customized separation oracle \cite{DBLP:conf/focs/DiakonikolasKK016}.
(Their optimization problem is implicit in the sense that it is defined in terms of the uncorrupted sample which we do not observe.)

An unusual feature of the aforementioned system of equations $\cE$ is that it also includes variables for the sum-of-squares proof of the inequality \cref{eq:subgaussian-vector} because we want to restrict the search to those sets $X'$ such that the uniform distribution $X'$ is certifiably subgaussian.
It is interesting to note that in this way we can use sum-of-squares as an approach to solve $\exists\, \forall$-problems as opposed to just the usual $\exists$-problems.
(The current problem is an $\exists\, \forall$-problem in the sense that we want to find $X'$ such that for all vectors $u$ the inequality \cref{eq:subgaussian-vector} holds for the uniform distribution over $X'$.)

We remark that the idea of using sum-of-squares to solve problems with quantifier alternation also plays a role in control theory (where the goal is find a dynamical system together with an associated Lyapunov functions, which can be viewed as sum-of-squares proof of the fact that the dynamical system behaves nicely in an appropriate sense).
However, to the best of our knowledge, this work is the first that uses this idea for the design of computationally-efficient algorithms with provable guarantees.
We remark that in a concurrent and independent work, Hopkins and Li use similar ideas to learn mixtures of well-separated spherical Gaussians \cite{HopkinsLi17}.
In a subsequent paper, Kothari and Steinhardt use those ideas for clustering \cite{KothariSteinhardt17}.

\subsection{Results}

Without any assumptions about the underlying distribution, the best known efficient algorithms for robust mean estimation incur an estimation error that depends on the spectral norm of the covariance matrix $\Sigma$ of the underlying distribution and is proportional to $\sqrt \e$ (where $\e>0$ is the fraction of outliers) \cite{DBLP:journals/corr/DiakonikolasKK017,DBLP:journals/corr/SteinhardtCV17}.
Concretely, given an $\e$-corrupted sample of sufficiently larger polynomial size from a distribution $D$, they compute an estimate $\hat \mu$ for the mean $\mu$ of $D$ such that with high probability $\norm{\hat \mu-\mu}\le O(\sqrt \e)\cdot \norm{\Sigma}^{1/2}$.
Furthermore, this bound is optimal for general distributions in the sense that up to constant factors no better bound is possible information-theoretically in terms of $\e$ and the spectral norm of $\Sigma$.

In the following theorem, we show that better bounds for the mean estimation error are possible for large classes of distributions.
Concretely, we assume (certifiable) bounds on higher-order moments of the distribution (degree 4 and higher).
These higher-order moment assumptions allow us to improve the estimation error as a function of $\e$ (instead of a $\sqrt \e$ bound as for the unconditional mean estimation before we obtain an $\e^{1-1/k}$ if we assume a bound on the degree-$k$ moments).
Furthermore, we also obtain multiplicative approximations for the covariance matrix (in the L\"owner order sense) regardless of the spectral norm of the covariance.
(Note that our notion of certifiable subgaussianity does not restrict the covariance matrix in any way.)

\begin{theorem}[Robust mean and covariance estimation under certifiable subgaussianity]
  \label{thm:mean-covariance-main}
  For every $C>0$ and even $k\in \N$, there exists a polynomial-time algorithm that given a (corrupted) sample $S\subseteq \R^d$ outputs a mean-estimate $\hat \mu\in \R^d$ and a covariance-estimate $\hat \Sigma\in \R^{d\times d}$ with the following guarantee:
  there exists $n_0\le (C+ d)^{O(k)}$
  such that if $S$ is an $\e$-corrupted sample with size $\card{S}\ge n_0$ of a $k$-certifiably $C$-subgaussian distribution $D$ over $\R^d$ with mean $\mu\in \R^d$ and covariance $\Sigma\in \R^{d\times d}$, then with high probability
  \begin{gather}
    \norm{\mu - \hat \mu} \le O(Ck)^{1/2}\cdot \e^{1-1/k} \cdot \norm{\Sigma}^{1/2}\\
    \norm{\Sigma^{-1/2}(\mu - \hat\mu)} \le O(Ck)^{1/2}\cdot \e^{1-1/k}\\
    (1-\delta)\Sigma \preceq \hat \Sigma \preceq (1+\delta)\Sigma \quad \text{for }\delta\le O(C k)\cdot \e^{1-2/k}\mper
  \end{gather}
  For the last two bounds, we assume in addition $Ck\cdot \e^{1-2/k}\le \Omega(1)$.\footnote{This notation means that we require $Ck\cdot \e^{1-2/k}\le c_0$ for some absolute constant $c_0>0$ (that could in principle be extracted from the proof).}
\end{theorem}

Note that the second guarantee for the mean estimation error $\mu-\hat \mu$ is stronger because $\norm{\mu-\hat \mu}\le \norm{\Sigma}^{1/2} \cdot \norm{\Sigma^{-1/2} (\mu-\hat \mu)}$.
We remark that $\norm{\Sigma^{-1/2}(\mu-\hat\mu)}$ is the Mahalanobis distance between $\hat \mu$ and $D$.

In general, our results provide information theoretically tight way to utilize higher moment information and give optimal dependence of the error on the fraction of outliers in the input sample for every $(k,\ell)$-certifiably $O(1)$-subgaussian distribution. This, in particular, improves on the mean estimation algorithm of \cite{DBLP:conf/focs/LaiRV16} by improving on their error bound of $O(\epsilon^{1/2})\|\Sigma\|^{1/2} \sqrt{\log{(d)}}$ to $O(\epsilon^{3/4}\|\Sigma\|^{1/2}$ under their assumption of bounded 4th moments (whenever ``certified'' by SoS).

\paragraph{Frobenius vs spectral norm for covariance-matrix estimation}

Previous work for robust covariance estimation \cite{DBLP:conf/focs/LaiRV16,DBLP:journals/corr/DiakonikolasKK017a} work with Frobenius norms for measuring the estimation error $\Sigma-\hat \Sigma$ and obtain in this way bounds that can be stronger than ours.
However, it turns out that assuming only $k$-certifiable subgaussianity makes it information-theoretically impossible to obtain dimension-free bounds in Frobenius norm and that we have to work with spectral norms instead. In this sense, the assumptions we make about distributions are substantially weaker compared to previous works.

Concretely, \cite{DBLP:conf/focs/LaiRV16} show a bound of $\normf{\Sigma-\hat\Sigma}\le \tilde O(\e^{1/2})\norm{\Sigma}$ assuming a $4$-th moment bound and that the distribution is an affine transformation of a product distribution.
\cite{DBLP:journals/corr/DiakonikolasKK017a} show a bound $\normf{\Sigma-\hat \Sigma}\le O(\e)\norm{\Sigma}$ assuming the distribution is Gaussian.\footnote{
  In the Gaussian case, \cite{DBLP:journals/corr/DiakonikolasKK017a} establish the stronger bound $\normf{\Sigma^{-1/2}(\Sigma-\hat \Sigma)\Sigma^{-1/2}}\le O(\e)$.
  This norm can be viewed as the Frobenius norm in a transformed space.
  Our bounds are for the spectral norm in the same transformed space.}

However, even for simple $k$-certifiably subgaussian distributions with parameter $C\le O(1)$, the information-theoretically optimal error in terms of Frobenius norm is $\normf{\Sigma-\hat\Sigma}\le O(\sqrt{d}\cdot \e^{1-2/k})\cdot \norm{\Sigma}$. Concretely, consider the mixture of $\cN(0,I)$ and $\cN(0, \epsilon^{-2/k}I)$ with weights $1-\epsilon$ and $\epsilon$, respectively. 
 Then, it's easy to confirm that both the mixture and the standard gaussian $\cN(0,I)$ are $k$-certifiably $2$-subgaussian (using Lemma \ref{lem:subgaussianity-of-gaussian} and Lemma \ref{lem:certifiably-subgaussian-mixtures}) and at most $\epsilon$ far in total variation distance. Thus, given only $k$-certifiable $2$-subgaussianity, it is information theoretically impossible to decide which of the above two distributions generated a given $\epsilon$-corrupted sample. Finally, the difference of their covariance equals $\epsilon (\epsilon^{-2/k}-1) I$ which has Frobenius norm $\Omega(\epsilon^{1-2/k}) \sqrt{d}$ for any $\epsilon < 1$.
 For this case, our algorithms guarantee the significantly stronger\footnotemark{} spectral-norm bound $\norm{\Sigma-\hat\Sigma}\le O(\e^{1-2/k})\cdot \norm{\Sigma}$.
\footnotetext{
  Similarly to \cite{DBLP:journals/corr/DiakonikolasKK017a}, we also work with norms in a transformed space (in which the distribution is isotropic) and obtain the stronger bound $\norm{\Sigma^{-1/2}(\Sigma-\hat\Sigma)\Sigma^{-1/2}}\le O(\e^{1-2/k})$.
}

\paragraph{Multiplicative vs. additive estimation error}
Another benefit of our covariance estimation algorithm is that provide a multiplicative approximation guarantee, that is, the quadratic form of the estimated covariance at any vector $u$ is within $(1 \pm \delta)$ of the quadratic form of the true covariance.
This strong guarantee comes in handy, for example, in \emph{whitening} or computing an isotropic transformation of the data---a widely used primitive in algorithm design.
Indeed, this ability to use the estimated covariance to whiten the data is crucial in our outlier-robust algorithm for independent component analysis (see \cref{sec:outl-robust-indep}).
The Frobenius norm error guarantees, in general, do not imply good multiplicative approximations and thus cannot be used for this application.

We note that in the special case of gaussian distributions, the results of Diakonikolas et al. \cite{DBLP:conf/focs/DiakonikolasKK016} allow recovering the mean and covariance in fixed polynomial time with better dependence of the error on the fraction of outliers that grows as $\tilde{O}(\epsilon) \|\Sigma\|^{1/2}$.
Our results when applied for the special case of gaussian mean and covariance estimation will require a running time of $d^{O(\sqrt{\log{(1/\epsilon)}})}$ to achieve a similar error guarantee.
Their algorithm for gaussian covariance estimation also provides multiplicative error guarantees.

\Dnote{}

\Pnote{}

\paragraph{Robust Estimation of Higher Moments}

Our techniques for robustly estimating mean and covariance of a distribution extend in a direct way to robustly estimating higher-order moment tensors.
In order to measure the estimation error, we use a variant of the injective tensor norm (with respect to a transformed space where the distribution is isotropic), which generalizes the norm we use for the estimation error of the covariance matrix.
This error bound allows us to estimate for every direction $u\in \R^d$, the low-degree moments of the distribution in direction $u$ with small error compared to the second moment in direction $u$.

The approaches in previous works fact inherent obstacles in generalizing to the problem of estimating the higher moments with multiplicative (i.e. in every direction $u$) error guarantees. This type of error is in fact crucial in applications for learning latent variable models
such as mixtures of Gaussians and independent component analysis. 

In fact, our guarantees are in some technical way stronger, which is crucial for our applications of higher-order moment estimates.
Unlike spectral norms, injective norms are NP-hard to compute (even approximately, under standard complexity assumptions).
For this reason, it is not clear how to make use of an injective-norm guarantee when processing moment-estimates further.
Fortunately, it turns out that our algorithm not only guarantees an injective-norm bound for the error but also a good certificate for this bound, in form of a low-degree sum-of-squares proof.
It turns out that this kind of certificate is precisely what we need for our applications---in particular, recent tensor decomposition algorithms based on sum-of-squares \cite{DBLP:conf/focs/MaSS16} can tolerate errors with small injective norm if that is certified by a low-degree sum-of-squares proof.

\begin{theorem}[Robust Higher Moment Estimation]
 \label[theorem]{thm:higher-moment-main}
 For every $C>0$ and even $k\in \N$, there exists a polynomial-time algorithm that given a (corrupted) sample $S\subseteq \R^d$ outputs a moment-estimates $\hat M_2\in \R^{d^2},\ldots,\hat M_k\in\R^{d^{k}}$ with the following guarantee:
  there exists $n_0\le (C+ d)^{O(k)}$   such that if $S$ is an $\e$-corrupted sample with size $\card{S}\ge n_0$ of a $k$-certifiably $C$-subgaussian distribution $D$ over $\R^d$ with moment-tensors $M_2\in\R^{d^2},\ldots,M_k\in\R^{d^k}$ such that $C k \cdot \e^{1-2/k}\le \Omega(1)$, then with high probability for every $r\le k/2$,
  \begin{equation}
    \label{eq:multiplicative-higher-moment}
    \forall u\in \R^d.~\iprod{M_{r} - \hat{M}_{r},u^{\otimes r}}^2
    \le \delta_r \cdot\iprod{M_{2},u^{\otimes 2}}^{r}
    \quad \text{ for } \delta_r\le O\Paren{ Ck}^{r/2}\cdot \e^{1-\frac{r}{k}}
  \end{equation}
  Furthermore, there exist degree-$k$ sum-of-squares proofs of the above polynomial inequalities in $u$.
\end{theorem}

\paragraph{Information-theoretic optimality}
We show in \cref{sec:lower-bounds} that the error guarantees in our robust moment-estimation algorithms are tight in their dependence on both $k$ and $\epsilon$.
For example, we show that there are two $k$-certifiably $O(1)$-subgaussian distributions with statistical distance $\e$ but means that are $\Omega(\sqrt{k} \epsilon^{1-1/k})$ apart.
A similar statement holds for higher-order moments.
The distributions are just mixtures of two one-dimensional Gaussians.

\paragraph{Application: independent component analysis}
As an immediate application of our robust moment estimation algorithm, we get an algorithm for Outlier Robust Independent Component Analysis. Independent component analysis (also known as blind source separation) is a fundamental problem in signal processing, machine learning and theoretical computer science with applications to diverse areas including neuroscience. Lathauwer et. al. \cite{MR2473563-DeLathauwer07}, following up on a long line of work gave algorithms for ICA based on 4th order tensor decomposition. A noise-tolerant version of this algorithm was developed in \cite{DBLP:conf/focs/MaSS16}. There is also a line of work in theoretical computer science on designing efficient algorithms for ICA \cite{DBLP:journals/corr/GoyalVX13,DBLP:conf/colt/VempalaX15}. 

In the ICA problem, we are given a non-singular\footnote{We could also consider the case that $A$ is rectangular and its columns are linearly independent.
Essentially the same algorithm and analysis would go through in this case.
We focus on the quadratic case for notational simplicity.} mixing matrix $A \in \R^{d \times d}$ with condition number $\kappa$ and a product distribution on $\R^d$. The observations come from the model $\{Ax\}$ that is, the observed samples are linear transformations (using the mixing matrix) of independent draws of the product random variable $x$. The goal is to recover columns of $A$ up to small relative error in Euclidean norm (up to signs and permutations) from samples. It turns out that information theoretic recovery of $A$ is possible whenever at most one source is non-gaussian. A widely used convention in this regard is the 4th moment assumption: for each $i$, $\E[ x_i^4] \neq 3\E[x_i^2]$. It turns out that we can assume $\E[x_i^2] = 1$ without loss of generality so this condition reduces to asserting $\E[x_i^4] \neq 3$. 

Outlier robust version of ICA was considered as an application of the outlier-robust mean and covariance estimation problems in \cite{DBLP:conf/focs/LaiRV16}. They sketched an algorithm with the guarantee that the relative error in the
columns of $A$ is at most $\epsilon \sqrt{\log{(d)}} \poly(\kappa)$.

In particular, this guarantee is meaningful only if the fraction of outliers $\epsilon \ll \frac{1}{\sqrt{\log{(d)}} \poly(\kappa)}$. Here, we improve upon their result by giving an outlier-robust ICA algorithm that recovers columns of $A$ up to an error that is \emph{independent} of both dimension $d$ and the condition number $\kappa$ of the mixing matrix $A$. 

Our algorithm directly follows by applying 4th order tensor decomposition. However, a crucial step in the algorithm involves ``whitening'' the 4th moments by using an estimate of the covariance matrix. Here, the multiplicative guarantees obtained in estimating the covariance matrix are crucial - estimates with respect to the Frobenius norm error do not give such whitening transformation in general. This whitening step essentially allows us to pretend that the mixing matrix $A$ is well-conditioned leading to no dependence on the condition number in the error.

\begin{theorem}[Robust independent component analysis]
  \label{thm:outlier-robust-ica}
  For every $C\ge 1$ and even $k\in \N$, there exists a polynomial-time algorithm that given a (corrupted) sample $S\subseteq \R^d$ outputs component estimates $\hat a_1,\ldots,\hat a_d\in \R^d$ with the following guarantees:
  Suppose $A \in \R^{d \times d}$ is a non-singular matrix with condition number $\kappa$ and columns $a_1,\ldots,a_d\in \R^d$.
  Suppose  $\mathbf x$ is a centered random vector with $d$ independent coordinates such that every coordinate $i\in [d]$ satisfies $\E[\mathbf x_i^2] = 1$, $\E[\mathbf x_i^4] -3 = \gamma \neq 0$, and  $\E[\mathbf x_i^{k}]^{1/k} \leq \sqrt{Ck}$.
  Then, if $S$ is an $\e$-corrupted sample of size $\card{S}\ge n_0$ from the distribution $\set{A \mathbf x}$, where $n_0\le (C+\kappa +d)^{O(k)}$, the component estimates satisfy with high probability
  \begin{equation}
    \max_{\pi \in S_d}\min_{i\in [d]} \iprod{A^{-1}\hat a_i, A^{-1} a_{\pi(i)}}^2 \ge 1-\delta
    \quad \text{ for } \delta < (1+\tfrac 1{\abs{\gamma}}) \cdot O(C^2 k^2) \cdot \epsilon^{1-4/k}
    \mper
  \end{equation}
\end{theorem} 

The quantity $\iprod{A^{-1}\hat a_i, A^{-1} a_{\pi(i)}}$ is closely related to the Mahalanobis distance between $\hat a_i$ and $a_{\pi(i)}$ with respect to the distribution $\set{A \mathbf x}$

\paragraph{Application: learning mixtures of Gaussians}

As yet another immediate application of our robust moment estimation algorithm, we get an outlier-robust algorithm for learning mixtures of spherical Gaussians.
Our algorithm works under the assumption that the means are linearly independent (and that the size of the sample grows with their condition number).
In return, our algorithm does not require the means of the Gaussians to be well-separated.
Our algorithm can be viewed as an outlier-robust version of tensor-decomposition based algorithms for mixtures of Gaussians \cite{MR3385380-Hsu13,DBLP:conf/stoc/BhaskaraCMV14}.

\begin{theorem}[Robust estimation of mixtures of spherical Gaussians (See Theorem \ref{thm:mixture-of-gaussian-technical} for detailed error bounds)]

Let $D$ be mixtures of $\cN(\mu_i,I)$ for $i \leq q$ with uniform\footnote{While our algorithm generalizes naturally to arbitrary mixture weights, we restrict to this situation for simplicity} mixture weights. Assume that $\mu_i$s are linearly independent and, further, assume that $\kappa$, the smallest non-zero eigenvalue of $\frac{1}{q}\sum_i \mu_i \mu_i^{\top}$ is $\Omega(1)$.

Given an $\epsilon$-corrupted sample of size $n \geq n_0= \Omega((d\log{(d)})^{k/2}/\epsilon^2)$, for every $k \geq 4$, there's a $\poly(n) d^{O(k)}$ time algorithm that recovers $\hat{\mu}_1, \hat{\mu}_2,\ldots \hat{\mu}_q$ so that there's a permutation $\pi:[q] \rightarrow [q]$ satisfying \[
\max_i \| (\frac{1}{q}\sum_i \mu_i \mu_i^{\top})^{-1/2}(\hat{\mu}_i - \mu_{\pi(i)})\| \leq O(qk)\epsilon^{1/3-1/k}.
\]

  \Dnote{}
  \Pnote{}
\end{theorem} 

Diakonikolas et. al. \cite{DBLP:conf/focs/DiakonikolasKK016} gave an outlier-robust algorithm that learns mixtures of $q$ gaussians with error $\approx q\sqrt{\epsilon}$ in each of the recovered means.
Their algorithm is polynomial in the dimension but has an exponential dependence on number of components $q$ in the running time.
Under the additional assumption that the means are linearly independent, our algorithm (say for $k=8$) recovers similar error guarantees as theirs but runs in time polynomial in both $q$ and $d$.
The key difference is the power of our algorithm to recover a multiplicative approximation to the 4th moment tensor which allows us to apply blackbox tensor decomposition based methods and run in fixed polynomial time \cite{DBLP:conf/innovations/HsuK13}.

%% file: content/identifiability.tex
\section{Robust identifiability of low-degree moments}
\label{sec:robust-identifiability}

In this section, we show that low-degree moments of certifiably subgaussian distributions are identifiable in the presence of outliers.
As we will formalize in \cref{sec:moment-estim-algor}, these proofs of identifiability are captured by the sum-of-squares proof system at low-degree.
This fact is the basis of our polynomial-time algorithms for robust moment estimation (also \cref{sec:moment-estim-algor}).

We recall the basic strategy for the identifiability proofs as described in the introduction:
Let $D_0$ be a certifiably subgaussian distribution whose moments we aim to estimate.
If we take a large enough sample $X\subseteq \R^d$ from $D_0$, then with high probability, the empirical distribution $D$ (uniform over the sample) has approximately the same low-degree moments as the (population) $D_0$.
In this case, the low-degree moments of $D$ are also certifiably subgaussian\footnote{
  Roughly speaking, certifiable subgaussianity is inherited by the empirical distribution from the population distribution because it is a property of the low-degree moments of distributions.
  We prove this claim in \cref{sec:subgaussianity}.
} (for essentially the same parameters $C$ as $D_0$).
Our input consists of an $\e$-corruption $Y\subseteq \R^d$ of the sample $X$.
In order to show that the low-degree moments of $D_0$ (or equivalently $D$) are identifiable from $Y$, we analyze the following (a-priori inefficient) estimation procedure:
\begin{enumerate}
\item
  find a certifiably subgaussian distribution $D'$ that has statistical distance at most $\e$ from the uniform distribution over $Y$,
\item
  output the low-degree moments of $D'$.
\end{enumerate}

We can typically find a distribution $D'$ as above, because with high-probability $D$ satisfies the conditions.
What remains to prove is that no matter what distribution $D'$ satisfying those conditions we choose, our moment estimates have low error.
To this end, we show several concrete and quantiative instances of the following general mathematical statement:
\begin{quote}
  \itshape
  If two distributions $D$ and $D'$, whose low-degree moments are (certifiably) subgaussian, have small statistical distance, then their low-degree moments are close.
\end{quote}

Our first bound of this kind, controls the distance between the means of $D$ and $D'$ in terms of their covariances $\Sigma$ and $\Sigma'$.
(Later bounds will also relate $\Sigma$ and $\Sigma'$.)

\begin{lemma}[Robust identifiability of mean]
  \label[lemma]{lem:mean-identifiability}
  Let $k\in \N$ be even.
  Let $D,D'$ be two distributions over $\R^d$ with means $\mu,\mu'\in \R^d$ and covariances $\Sigma,\Sigma'\in \R^{d\times d}$.
  Suppose  $\normtv{D-D'}\le \e<0.9$ and that $D$ and $D'$ are $(k,\ell)$-certifiably subgaussian with parameter $C>0$.
  Then, 
  \begin{equation}
    \label{eq:mean-identifiability}
    \forall u\in \R^d.~\abs{\iprod{u,\mu - \mu'}}
    \le \delta
    \cdot \iprod{u, (\Sigma+\Sigma') u}^{1/2}\quad \text{ where }  \delta= O\Paren{\sqrt{C k\,}\cdot \e^{1-1/k}}\mper
  \end{equation}
  In particular, $\norm{\mu -\mu'}\le \delta \norm{\Sigma+\Sigma'}^{1/2}$ and $\norm{(\Sigma+\Sigma')^{-1/2}(\mu-\mu')}\le \delta$.
  (Here, $\R^d$ is endowed with the standard Euclidean norm and $\R^{d\times d}$ with the corresponding operator norm.)
\end{lemma}
\begin{proof}
  Consider a coupling between the distributions $D(x)$ and $D'(x')$ such that $\Pr\set{x=x'}\ge 1-\e$.
  By \Holder's inequality, we have for every vector $u\in \R^d$,
  \begin{align}
     \abs{\iprod{u, \mu-\mu'}}= \abs{\E \iprod{u, x- x'} \Ind_{x\neq x'}}
    \le \Paren{\E \iprod{u, x- x'}^k}^{1/k}
    \underbrace{\Paren{\E \Ind_{x\neq x'}^{k/(k-1)}}^{1-1/k}}_{\le \e^{1-1/k}}
  \end{align}
  At the same time,
  \begin{align}
    \Paren{\E \iprod{u,x-x'}^{k}}^{1/k}
    & = \Paren{\E \iprod{u,x-\mu-(x'-\mu') + (\mu-\mu')}^{k}}^{1/k}\\
    & \le
      \begin{aligned}[t]
        & \Paren{\E \iprod{u,x-\mu}^{k}}^{1/k}
        +\Paren{\E \iprod{u,x'-\mu'}^{k}}^{1/k}
        +\abs{\iprod{u,\mu-\mu'}} \\
        & \text{(using triangle inequality for $L_k$-norm)}
      \end{aligned}\\
    & \le
      \begin{aligned}[t]
        & \sqrt{Ck\,} \cdot \Paren{ \bigparen{\E \iprod{u,x-\mu}^2}^{1/2}
          + \bigparen{\E \iprod{u,x'-\mu'}^2}^{1/2}}
        + \abs{\iprod{u,\mu-\mu'}}\\
        & \text{(using certifiable subgaussianity)} 
      \end{aligned}\\
    & \le O(\sqrt{C k\,})\cdot
      \iprod{u,(\Sigma+\Sigma')u}^{1/2}
      + \abs{\iprod{u,\mu-\mu'}}
      \mper
  \end{align}
  Combining the two bounds yields 
  \begin{displaymath}
    \abs{\iprod{u,\mu-\mu'}} \le \e^{1-1/k} \cdot \
    \Paren{ O(\sqrt{C k\,}) \cdot \iprod{u,(\Sigma+\Sigma')u}^{1/2}
      + \abs{\iprod{u,\mu-\mu'}}}\mper
  \end{displaymath}
  Since we assume that $\e<0.9$, we can conclude that $\abs{\iprod{u,\mu-\mu'}}\le \delta\cdot \iprod{u,(\Sigma+\Sigma')u}^{1/2}$ for $\delta\le O\Paren{\e^{1-1/k}\cdot \sqrt{C k\,}}$ as desired.
\end{proof}

Next we show that if two certifiably subgaussian distributions have small statistical distance, then their (raw) second moments are close in the \Lowner order sense.

\begin{lemma}[Robust identifiability of second moment]
  \label[lemma]{lem:multiplicative-second-moment}
  Let $D,D'$ be two distributions over $\R^d$ with second-moment matrices $M,M'\in \R^{d\times d}$.
  Suppose $\normtv{D-D'}\le \e<0.9$ and that $D$ and $D'$ are $(k,\ell)$-certifiably subgaussian with parameters $C>0$.
  Then,
  \begin{equation}
    \label{eq:multiplicative-second-moment}
    -\delta \cdot (M+M')\preceq M - M' \preceq \delta \cdot (M+M')
    \quad \text{for } \delta\le O\Paren{ Ck \cdot \e^{1-2/k}}
  \end{equation}
In particular, $\bigl(1-O(\delta)\bigr) M' \preceq M \preceq \bigl(1+O(\delta)\bigr) M'$ if $\delta<0.9$.
\end{lemma}

\begin{proof}
  Consider a coupling between the distributions $D(x)$ and $D'(x')$ such that $\Pr\set{x=x'}\ge 1-\e$.
  Then, for every vector $u\in \R^d$,
  \begin{align*}
    \abs{\iprod{u, (M-M') u}}
    & = \abs{\E \iprod{u,x}^2 -\iprod{u,x'}^2}\\
    & \le \e^{1-2/k} \cdot \Paren{\E \Paren{\iprod{u,x}^2-\iprod{u,x'}^2}^{k/2}}^{2/k}
    \quad \text{(using \Holder inequality)}\\
    & \le \e^{1-2/k} \cdot \Paren{
      \Paren{\E\iprod{u,x}^{k}}^{2/k}
      + \Paren{\E\iprod{u,x'}^{k}}^{2/k}}
      \quad\text{(by $L_{k/2}$ triangle inequ.)}\\
    & \le \e^{1-2/k} \cdot O(C k) \cdot \Paren{
      \E\iprod{u,x}^{2}
      + \E\iprod{u,x'}^{2}}
    \quad\text{(by subgaussianity)}\\
    & = O( C k \cdot \e^{1-2/k}) \cdot \iprod{u,(M+M')u} \mcom
  \end{align*}
  which implies the desired bound $-\delta \cdot (M+M')\preceq M - M' \preceq \delta \cdot (M+M')$ for $\delta\le O\Paren{ Ck \cdot \e^{1-1/k}}$.
  For the third inequality, we use \cref{lem:shifts-of-subgaussian-are-subgaussian} (moment bounds for shifts of subgaussian distributions).

  For $\delta\le 1$, we can now rearrange this bound to get the following relationship between $M$ and $M'$,
  \begin{equation}
    \tfrac{1-\delta}{1+\delta}\cdot M' \preceq M  \preceq \tfrac{1+\delta}{1-\delta}\cdot M'\mcom
  \end{equation}
  which implies the desired bound $\bigl(1-O(\delta)\bigr) M' \preceq M \preceq \bigl(1+O(\delta)\bigr) M'$ for $\delta<0.9$.
\end{proof}

Next, we combine the previous bounds for the first and second moments in order to establish the identifiability of the covariance matrix.

\begin{corollary}[Covariance identifiability]
  \label[corollary]{cor:multiplicative-covariance-identifiability}
  Under the same conditions as the previous \cref{lem:multiplicative-second-moment},
  the covariance matrices $\Sigma$ and $\Sigma'$ of the distributions $D$ and $D'$ satisfy,
  \begin{equation}
    \label{eq:multiplicative-covariance}
    -\delta \cdot (\Sigma+\Sigma')\preceq \Sigma - \Sigma' \preceq \delta \cdot (\Sigma+\Sigma')
    \quad \text{for } \delta\le O( Ck \cdot \e^{1-2/k})\mper
  \end{equation}
\end{corollary}

As before, in the case that $\delta\le 0.9$, the above bounds imply a strong multiplicative approximation of the covariance so that $(1-\delta')\cdot \Sigma' \preceq \Sigma \preceq (1+\delta')\Sigma'$ for some $\delta'\le O(\delta)$.

\begin{proof}[Proof of \cref{cor:multiplicative-covariance-identifiability}]
  Let $D_\mu,D'_\mu$ be the distributions $D,D'$ shifted by $\mu$ so that $D_\mu$ has expectation $0$ and $D'_\mu$ has expectation $\mu'-\mu$.
  Note that $\normtv{D_\mu-D'_{\mu}}=\normtv{D-D'}\le \e$.
  By \cref{lem:multiplicative-second-moment} (second-moment identifiability applied to $D_\mu,D_\mu'$), every vector $u\in \R^d$ satisfies,
  \begin{equation}
    \abs{\iprod{u,(\Sigma-\Sigma'-\dyad{(\mu-\mu')})u}}
    \le \delta_1 \cdot \iprod{u,(\Sigma+\Sigma'+\dyad{(\mu-\mu')})u}\mper
  \end{equation}
  Here, $\delta_1\le O(Ck\cdot \e^{1-1/k})$.
  By \cref{lem:mean-identifiability} (mean identifiability),
  $\iprod{u,\dyad{(\mu-\mu')}u}\le \delta_2 \cdot \iprod{u,(\Sigma+\Sigma')u}$, where $\delta_2\le O(Ck\cdot \e^{1-1/k})$.
  Combining the inequalities gives the desired bound.
\end{proof}

Taking together \cref{cor:multiplicative-covariance-identifiability} and \cref{lem:mean-identifiability}, we get a mean estimation error that depends only on the covariance of $D$ (as opposed to the covariances of both $D$ and $D'$).

\begin{corollary}[Stronger mean identifiability]
  \label[corollary]{cor:mean-identifiability-no-covariance-bound}
  Under the same conditions as the previous \cref{lem:multiplicative-second-moment} and the additional condition $C k \cdot \e^{1-2k}\le \Omega(1)$, the means $\mu$ and $\mu'$  of the distributions $D$ and $D'$ satisfy,
  \begin{equation}
    \label{eq:mean-identifiability-final}
    \forall u\in \R^d.~\abs{\iprod{u,\mu - \mu'}}
    \le \delta
    \cdot \iprod{u, \Sigma u}^{1/2}\quad \text{ where } \delta= O\Paren{\sqrt{C k\,}\cdot \e^{1-1/k}}\mcom
  \end{equation}
  where $\Sigma\in\R^{d\times d}$ is the covariance of $D$.
  In particular, $\norm{\mu -\mu'}\le \delta \norm{\Sigma}^{1/2}$ and $\norm{\Sigma^{-1/2}(\mu-\mu')}\le \delta$.
  (Here, $\R^d$ is endowed with the standard Euclidean norm and $\R^{d\times d}$ with the corresponding operator norm.)
\end{corollary}

Next we bound the distances of higher order moments for certifiably subgaussian distributions that are close in statistical distance.
Compared to first and second order moments, distances for higher order moments are less standard.
Our choice for the distance is informed by the applications of independent component analysis and learning mixtures of Gaussians.
We view the order-$2r$ moments of $D$ and $D'$ as two polynomials $p$ and $p'$ of degree $2r$ and we bound their difference $p(u)-p'(u)$ relative to the variance of the distribution in direction $u$.

\begin{lemma}[Robust identifiability of higher moments]
  \label[lemma]{lem:multiplicative-higher-moment}
  Let $r\in \N$ and let $D,D'$ be two distributions over $\R^d$ with second moments $M_2,M'_2\in \R^{d^2}$  and $r$-th moments $M_{r},M'_{r}\in \R^{d^{r}}$.
  Suppose $\normtv{D-D'}\le \e<0.9$ and that $D$ and $D'$ are $(k,\ell)$-certifiably subgaussian with parameters $C>0$ and $2r \le k$.
  Then, for every $u \in \R^d$,
  \begin{equation}
    \label{eq:multiplicative-higher-moment-identifiability}
    \abs{\iprod{M_{r} - M_{r}',u^{\otimes r}}} 
    \le \delta \cdot \iprod{M_{2}+M_2',u^{\otimes 2}}^{r/2}
    \quad \text{for } \delta\le O(Ck)^{r/2} \cdot \e^{1-\frac{r}{k}}
  \end{equation}
\end{lemma}

\begin{proof}
Consider a coupling between the distributions $D(x)$ and $D'(x')$ such that $\Pr\set{x=x'}\ge 1-\e$.
  Then, for every vector $u\in \R^d$,
  \begin{align*}
    \abs{\iprod{u, (M_{r}-M_{r}') u}}
    & = \abs{\E \iprod{u,x}^{r} -\iprod{u,x'}^{r}}\\
    & \le \e^{1-r/k} \cdot \Paren{\E \Paren{\iprod{u,x}^{r}-\iprod{u,x'}^{r}}^{k/r}}^{r/k}
    \quad \text{(using \Holder's inequality)}\\
    & \le \e^{1-r/k} \cdot \Paren{
      \Paren{\E\iprod{u,x}^{k}}^{r/k}
      + \Paren{\E\iprod{u,x'}^{k}}^{r/k}}
      \quad\text{(by $L_{k/r}$ triangle inequ.)}\\
    & \le \e^{1-r/k} \cdot O(C^r k^r) \cdot \Paren{
      \Paren{\E\iprod{u,x}^{2}}^{r/2}
      + \Paren{\E\iprod{u,x'}^{2}}^{r/2}}
    \quad\text{(by subgaussianity)}
  \end{align*}
  which implies the desired bound for $\delta\le \delta\le O(Ck)^{r/2} \cdot \e^{1-\frac{r}{k}}$.
  For the third inequality, we rely on \cref{lem:subgaussian-shift-moment-bound} (moment bounds for shifts of subgaussian distributions).
 \end{proof}

 \Dnote{}

%% file: content/preliminaries.tex
\section{Preliminaries}
\label{sec:preliminaries}

In this section, we define pseudo-distributions and sum-of-squares proofs.
See the lecture notes \cite{BarakS16} for more details and the appendix in \cite{DBLP:journals/corr/MaSS16} for proofs of the propositions appearing here.

Let $x = (x_1, x_2, \ldots, x_n)$ be a tuple of $n$ indeterminates and let $\R[x]$ be the set of polynomials with real coefficients and indeterminates $x_1,\ldots,x_n$.
We say that a polynomial $p\in \R[x]$ is a \emph{sum-of-squares (sos)} if there are polynomials $q_1,\ldots,q_r$ such that $p=q_1^2 + \cdots + q_r^2$.

\subsection{Pseudo-distributions}

Pseudo-distributions are generalizations of probability distributions.
We can represent a discrete (i.e., finitely supported) probability distribution over $\R^n$ by its probability mass function $D\from \R^n \to \R$ such that $D \geq 0$ and $\sum_{x \in \mathrm{supp}(D)} D(x) = 1$.
Similarly, we can describe a pseudo-distribution by its mass function.
Here, we relax the constraint $D\ge 0$ and only require that $D$ passes certain low-degree non-negativity tests.

Concretely, a \emph{level-$\ell$ pseudo-distribution} is a finitely-supported function $D:\R^n \rightarrow \R$ such that $\sum_{x} D(x) = 1$ and $\sum_{x} D(x) f(x)^2 \geq 0$ for every polynomial $f$ of degree at most $\ell/2$.
(Here, the summations are over the support of $D$.)
A straightforward polynomial-interpolation argument shows that every level-$\infty$-pseudo distribution satisfies $D\ge 0$ and is thus an actual probability distribution.
We define the \emph{pseudo-expectation} of a function $f$ on $\R^d$ with respect to a pseudo-distribution $D$, denoted $\pE_{D(x)} f(x)$, as
\begin{equation}
  \pE_{D(x)} f(x) = \sum_{x} D(x) f(x) \,\mper
\end{equation}
The degree-$\ell$ moment tensor of a pseudo-distribution $D$ is the tensor $\E_{D(x)} (1,x_1, x_2,\ldots, x_n)^{\otimes \ell}$.
In particular, the moment tensor has an entry corresponding to the pseudo-expectation of all monomials of degree at most $\ell$ in $x$.
The set of all degree-$\ell$ moment tensors of probability distribution is a convex set.
Similarly, the set of all degree-$\ell$ moment tensors of degree $d$ pseudo-distributions is also convex.
Key to the algorithmic utility of pseudo-distributions is the fact that while there can be no efficient separation oracle for the convex set of all degree-$\ell$ moment tensors of an actual probability distribution, there's a separation oracle running in time $n^{O(\ell)}$ for the convex set of the degree-$\ell$ moment tensors of all level-$\ell$ pseudodistributions.

\begin{fact}[\cite{MR939596-Shor87,parrilo2000structured,MR1748764-Nesterov00,MR1846160-Lasserre01}]
  \label[fact]{fact:sos-separation-efficient}
  For any $n,\ell \in \N$, the following set has a $n^{O(\ell)}$-time weak separation oracle (in the sense of \cite{MR625550-Grotschel81}):
  \begin{equation}
    \Set{ \pE_{D(x)} (1,x_1, x_2, \ldots, x_n)^{\otimes d} \mid \text{ degree-d pseudo-distribution $D$ over $\R^n$}}\,\mper
  \end{equation}
\end{fact}
This fact, together with the equivalence of weak separation and optimization \cite{MR625550-Grotschel81} allows us to efficiently optimize over pseudo-distributions (approximately)---this algorithm is referred to as the sum-of-squares algorithm.

The \emph{level-$\ell$ sum-of-squares algorithm} optimizes over the space of all level-$\ell$ pseudo-distributions that satisfy a given set of polynomial constraints---we formally define this next.

\begin{definition}[Constrained pseudo-distributions]
  Let $D$ be a level-$\ell$ pseudo-distribution over $\R^n$.
  Let $\cA = \{f_1\ge 0, f_2\ge 0, \ldots, f_m\ge 0\}$ be a system of $m$ polynomial inequality constraints.
  We say that \emph{$D$ satisfies the system of constraints $\cA$ at degree $r$}, denoted $D \sdtstile{r}{} \cA$, if for every $S\subseteq[m]$ and every sum-of-squares polynomial $h$ with $\deg h + \sum_{i\in S} \max\set{\deg f_i,r}$,
  \begin{displaymath}
    \pE_{D} h \cdot \prod _{i\in S}f_i  \ge 0\,.
  \end{displaymath}
  We write $D \sdtstile{}{} \cA$ (without specifying the degree) if $D \sdtstile{0}{} \cA$ holds.
  Furthermore, we say that $D\sdtstile{r}{}\cA$ holds \emph{approximately} if the above inequalities are satisfied up to an error of $2^{-n^\ell}\cdot \norm{h}\cdot\prod_{i\in S}\norm{f_i}$, where $\norm{\cdot}$ denotes the Euclidean norm\footnote{The choice of norm is not important here because the factor $2^{-n^\ell}$ swamps the effects of choosing another norm.} of the cofficients of a polynomial in the monomial basis.
\end{definition}

We remark that if $D$ is an actual (discrete) probability distribution, then we have  $D\sdtstile{}{}\cA$ if and only if $D$ is supported on solutions to the constraints $\cA$.

We say that a system $\cA$ of polynomial constraints is \emph{explicitly bounded} if it contains a constraint of the form $\{ \|x\|^2 \leq M\}$.
The following fact is a consequence of \cref{fact:sos-separation-efficient} and \cite{MR625550-Grotschel81},

\begin{fact}[Efficient Optimization over Pseudo-distributions]
There exists an $(n+ m)^{O(\ell)} $-time algorithm that, given any explicitly bounded and satisfiable system\footnote{Here, we assume that the bitcomplexity of the constraints in $\cA$ is $(n+m)^{O(1)}$.} $\cA$ of $m$ polynomial constraints in $n$ variables, outputs a level-$\ell$ pseudo-distribution that satisfies $\cA$ approximately. 
\end{fact}

\subsection{Sum-of-squares proofs}

Let $f_1, f_2, \ldots, f_r$ and $g$ be multivariate polynomials in $x$.
A \emph{sum-of-squares proof} that the constraints $\{f_1 \geq 0, \ldots, f_m \geq 0\}$ imply the constraint $\{g \geq 0\}$ consists of  polynomials $(p_S)_{S \subseteq [m]}$ such that
\begin{equation}
g = \sum_{S \subseteq [m]} p_S \cdot \Pi_{i \in S} f_i
\mper
\end{equation}
We say that this proof has \emph{degree $\ell$} if for every set $S \subseteq [m]$, the polynomial $p_S \Pi_{i \in S} f_i$ has degree at most $\ell$.
If there is a degree $\ell$ SoS proof that $\{f_i \geq 0 \mid i \leq r\}$ implies $\{g \geq 0\}$, we write:
\begin{equation}
  \{f_i \geq 0 \mid i \leq r\} \sststile{\ell}{}\{g \geq 0\}
  \mper
\end{equation}

Sum-of-squares proofs satisfy the following inference rules.
For all polynomials $f,g\colon\R^n \to \R$ and for all functions $F\colon \R^n \to \R^m$, $G\colon \R^n \to \R^k$, $H\colon \R^{p} \to \R^n$ such that each of the coordinates of the outputs are polynomials of the inputs, we have:

\begin{align}
&\frac{\cA \sststile{\ell}{} \{f \geq 0, g \geq 0 \} } {\cA \sststile{\ell}{} \{f + g \geq 0\}}, \frac{\cA \sststile{\ell}{} \{f \geq 0\}, \cA \sststile{\ell'}{} \{g \geq 0\}} {\cA \sststile{\ell+\ell'}{} \{f \cdot g \geq 0\}} \tag{addition and multiplication}\\
&\frac{\cA \sststile{\ell}{} \cB, \cB \sststile{\ell'}{} C}{\cA \sststile{\ell \cdot \ell'}{} C}  \tag{transitivity}\\
&\frac{\{F \geq 0\} \sststile{\ell}{} \{G \geq 0\}}{\{F(H) \geq 0\} \sststile{\ell \cdot \deg(H)} {} \{G(H) \geq 0\}} \tag{substitution}\mper
\end{align}

Low-degree sum-of-squares proofs are sound and complete if we take low-level pseudo-distributions as models.

Concretely, sum-of-squares proofs allow us to deduce properties of pseudo-distributions that satisfy some constraints.

\begin{fact}[Soundness]
  \label{fact:sos-soundness}
  If $D \sdtstile{r}{} \cA$ for a level-$\ell$ pseudo-distribution $D$ and there exists a sum-of-squares proof $\cA \sststile{r'}{} \cB$, then $D \sdtstile{r\cdot r'+r'}{} \cB$.
\end{fact}

If the pseudo-distribution $D$ satisfies $\cA$ only approximately, soundness continues to hold if we require an upper bound on the bit-complexity of the sum-of-squares $\cA \sststile{r'}{} B$  (number of bits required to write down the proof).

In our applications, the bit complexity of all sum of squares proofs will be $n^{O(\ell)}$ (assuming that all numbers in the input have bit complexity $n^{O(1)}$).
This bound suffices in order to argue about pseudo-distributions that satisfy polynomial constraints approximately.

The following fact shows that every property of low-level pseudo-distributions can be derived by low-degree sum-of-squares proofs.

\begin{fact}[Completeness]
  \label{fact:sos-completeness}
  Suppose $d \geq r' \geq r$ and $\cA$ is a collection of polynomial constraints with degree at most $r$, and $\cA \vdash \{ \sum_{i = 1}^n x_i^2 \leq B\}$ for some finite $B$.

  Let $\{g \geq 0 \}$ be a polynomial constraint.
  If every degree-$d$ pseudo-distribution that satisfies $D \sdtstile{r}{} \cA$ also satisfies $D \sdtstile{r'}{} \{g \geq 0 \}$, then for every $\epsilon > 0$, there is a sum-of-squares proof $\cA \sststile{d}{} \{g \geq - \epsilon \}$.
\end{fact}

%% file: content/algorithm.tex
\section{Moment estimation algorithm}
\label{sec:moment-estim-algor}

Our robust moment estimation algorithm is a direct translation of the robust identifiability results from Section \cref{sec:robust-identifiability}. 

Let $Y=\set{y_1,\ldots,y_n}\subseteq \R^d$ be an $\e$-corrupted sample of a distribution $D$ over $\R^d$.
We consider the following system $\cA\seteq\cA_{Y,\e}$ of quadratic equations in scalar-valued variables $w_1,\ldots,w_n$ and vector-valued variables $x'_1,\ldots,x'_n$,
\begin{equation}
  \cA_{Y,\e}\colon
  \left \{
    \begin{aligned}
      &&
      \textstyle\sum_{i=1}^n w_i
      &= (1-\e) \cdot n\\
      &\forall i\in [n].
      & w_i^2
      & =w_i \\
      &\forall i\in [n].
      & w_i \cdot (y_i - x'_i)
      & = 0
    \end{aligned}
  \right \}
\end{equation}

The following fact shows that the solutions to $\cA_{Y,\e}$ correspond to all $\e$-corruptions of the (already corrupted) sample $Y$.
In particular, if $Y=\set{y_1,\ldots,y_n}$ is an $\e$-corruption of a sample $X=\set{x_1,\ldots,x_n}$ of $D$, then one solution to $\cA$ consists of the uncorrupted sample $X$ so that $x'_1=x_1,\ldots,x'_n=x_n$.

\begin{fact}
  An assignment to the variables $x'_1,\ldots,x'_n$ can be extended to a solution for $\cA_{Y,\e}$ if and only if $x'_i=y_i$ holds for all but at most an $\e$-fraction of the indices $i\in [n]$.  
\end{fact}

Note that in order to extend such an assignment, we can choose $w_i=1$ for $(1-\e)n$ indices $i\in [n]$ such that $x_i=y_i$ and $w_i=0$ for the remaining indices.

Furthermore, we consider the following system of polynomial equations $\cB\seteq\cB_{C,k}$ in vector-valued variables $x'_1,\ldots,x'_n$, vector-valued variables $p_1,\ldots,p_{k/2}$, and matrix-valued variables $Q_1,\ldots,Q_{k/2}$.
The idea is that $p_1,\ldots,p_{k/2}$ and $Q_1,\ldots,Q_{k/2}$ are degree-$\ell$ sum-of-squares proofs of the inequalities \cref{eq:low-degree-subgaussian} ($C$-subgaussianity of moments up to degree $k$) for the uniform distribution over $x'_1,\ldots,x'_n$.
\begin{equation}
  \label{eq:sos-subgaussian-constraints}
  \cB_{C,k,\ell}\colon
  \left\{
  \begin{aligned}
    \forall k'\in \N_{\le k/2}.~
    \tfrac 1 n \sum_{i=1}^n \iprod{x'_i,u}^{2k'}
    = \Bigparen{Ck'\cdot\tfrac 1 n \sum_{i=1}^n \iprod{x'_i,u}^{2}}^{k'}
    & + \iprod{p_{k'},(1,u)^{\otimes \ell}}\cdot (1-\norm{u}^2)\\
    & - \norm{Q_{k'} \cdot (1,u)^{\otimes \ell/2}}^2
  \end{aligned}
  \right\}
\end{equation}
We emphasize that in the above system of equations the vector-valued variable $u=(u_1,\ldots,u_d)$ is not part of the solution (i.e., we are not seeking an assignment for it).
Instead, each equation in $\cB_{C,k,\ell}$ is about two polynomials with indeterminates $u_1,\ldots,u_d$, whose coefficients are polynomials in the variables $x'_1,\ldots,x'_n$, $p_1,\ldots,p_{k/2}$, and $Q_1,\ldots,Q_{k/2}$.
In particular, we could eliminate the variable $u$ from the equations in $\cB_{C,k,\ell}$ by explicitly writing down the expressions for the coefficients of the polynomials in $u$.
We refrain from doing so because these expressions would be rather complicated and would obfuscate the meaning of the equations.

The following fact shows that the solutions to the system $\cB_{C,k}$ correspond to certifiably subgaussian distributions.

\begin{fact}
  An assignment to the vector-valued variables $x'_1,\ldots,x'_n$ can be extended to a solution for $\cB_{C,k,\ell}$ if and only if the uniform distribution over $x'_1,\ldots,x'_n$ is $(k,\ell)$-certifiably subgaussian with parameter $C$.
\end{fact}

The following algorithm is the key ingredient of our efficient outlier-robust estimators.

\begin{mdframed}
  \begin{algorithm}[Algorithm for moment estimation via sum-of-squares]
    \label[algorithm]{alg:moment-estimation-program}\mbox{}
    \begin{description}
    \item[Given:]
      $\e$-corrupted sample $Y=\set{y_1,\ldots,y_n}\subseteq \R^d$ of a $(k,\ell)$-certifiably $0.9C$-subgaussian distribution $D_0$ over $\R^d$
    \item[Estimate:]
      moments $M_1,\ldots,M_k$ with $M_r=\E_{D_0(x)}x^{\otimes r}$
    \item[Operation:]\mbox{}
      \begin{enumerate}
      \item 
        find a level-$\ell$ pseudo-distribution $D'$ that satisfies $\cA_{Y,\e}$ and $\cB_{C,k,\ell}$
      \item
        output moment-estimates $\hat M_1,\ldots,\hat M_k$ with $\hat M_r = \pE_{D'(x')} \Brac{\tfrac 1 n \sum_{i=1}^n (x'_i)^{\otimes r}}$
      \end{enumerate}
    \end{description}    
  \end{algorithm}
\end{mdframed}

\Dnote{}

The following lemma is the sum-of-squares analog of \cref{lem:mean-identifiability}.%

\begin{lemma}
  \label[lemma]{lem:mean-estimation-sos}
  Let $D$ be the uniform distribution over $X \subseteq \R^d$ with mean $\mu\in \R^d$ and covariance $\Sigma \in \R^{d\times d}$.
  Suppose $D$ is $(k,\ell)$-certifiably subgaussian with parameter $C>0$ and that $Ck \cdot \e^{1-2/k}\le \Omega(1)$.\footnote{
    This notation means that we require $C k \e^{1-2/k} \le \gamma$ for some absolute constant $\gamma>0$, e.g., $\gamma=10^{-10}$.}
  Let $Y$ be an $\epsilon$-corruption of $X$ with $\e<0.01$. 
  Let $\mu',\Sigma'$ be the polynomials $\mu'(X')=\tfrac 1n \sum_{i=1}^n x'_i$ and $\Sigma'(X')=\tfrac 1n \sum_{i=1}^n \dyad{(x'_i)} - \dyad{(\mu')}$ in the vector-valued variables $x'_1,\ldots,x'_n$.
  Then, for every vector $u\in \R^d$,
  \begin{equation}
    \cA_{Y,\e} \cup \cB_{C,k,\ell}
    \sststile{\ell}{X'} \Set{ \iprod{u,\mu-\mu'(X')}^{k} \le \delta^k \iprod{u,(\Sigma + \Sigma'(X')) u}^{k/2}}\,,
  \end{equation}
  where $\delta= O\Paren{\sqrt{C k\,}\cdot \e^{1-1/k}}$.
  Furthermore, the bit complexity of the sum-of-squares proof is polynomial.
\end{lemma}

\begin{proof}
  Choose $r_1,\ldots,r_n\in \set{0,1}$ such that $\sum_i r_i=(1-\e)n$ and $r_i(x_i - y_i) = 0$ for all $i\in [n]$.
  (Such a choice exists because $Y$ is an $\epsilon$-corruption of $X$.)

  Letting $z_i=r_i\cdot w_i$, we claim that 
  \begin{displaymath}
    \cA_{Y,\e}\sststile{2}{w}
    \Set{\sum_{i=1}^n z_i \ge (1-2 \epsilon)n}
    \,.
  \end{displaymath}
  Indeed, since $\set{w_i^2=w_i} \sststile{2}{w_i} \set{(1-z_i)\le (1-w_i)+(1-r_i)}$, we have $\cA_{Y,\e}\sststile{2}{} \set{\sum_{i=1}^n (1-z_i)\le 2\e n}$

  By the sum-of-squares version of \Holder's inequality, we have for every vector $u\in \R^d$, using that $\cA_{Y,\e}\sststile{k}{} \set{(1-z_i)^{k-1} = (1-z_i)}$.
  \begin{align}
    \cA_{Y,\e} \cup \cB_{C,k,\ell} \sststile{\ell}{w,x'}
    \begin{aligned}[t]
      \Bigl \{
      \iprod{u, \mu-\mu'}^k
      &= \Bigparen{\tfrac{1}{n} \sum_i \iprod{u, x_i- x_i'} (1-z_i)}^k\\
      & \le \Bigparen{\tfrac{1}{n} \sum_i \iprod{u, x_i- x_i'}^k} \Bigparen{\tfrac{1}{n} \sum_i (1-z_i)^{k}}^{k-1}\\
      &\leq \Bigparen{\tfrac{1}{n} \sum_i \iprod{u, x_i- x_i'}^k} (2\epsilon)^{k-1}
      \Bigr \}
      \,.
    \end{aligned}
  \end{align}
  For brevity, we will use $\E_{i\in [n]}$ to denote the average over indices $i\in[n]$.
  At the same time, using the fact that $\sststile{k}{a,b,c} \Set{(a+b+c)^k \leq 3^k(a^k + b^k + c^k)}$ (similar to \cref{lem:binomial}),
  \begin{align}
    \cA_{Y,\e} \cup \cB_{C,k,\ell} \sststile{\ell}{w,x'}
    \begin{aligned}[t]
      \Bigl \{
    &\Bigparen{\E_{i\in[n]}  \iprod{u,x_i-x_i'}^{k}} \\
    & = \Bigparen{\E_{i\in[n]}  \iprod{u,x_i-\mu-(x_i'-\mu') + (\mu-\mu')}^{k}}\\
    & \le 3^k\Bigparen{\tfrac{1}{n} \sum_i \iprod{u,x_i-\mu}^{k}}
        +3^k \Bigparen{\tfrac{1}{n} \sum_i \iprod{u,x_i'-\mu'}^{k}}
        +3^k \iprod{u,\mu-\mu'}^k\\
    & \le
      \begin{aligned}[t]
        & 3^k \bigparen{Ck\E_{i\in[n]} \iprod{u,x_i-\mu}^2}^{k/2}
          + 3^k\bigparen{Ck\E_{i\in[n]} \iprod{u,x'_i-\mu'}^2}^{k/2}
          + 3^k \iprod{u,\mu-\mu'}^k\\
        & \text{(using certifiable subgaussianity)} 
      \end{aligned}\\
    & \le O(C k)^{k/2}\cdot
      \iprod{u,(\Sigma+\Sigma')u}^{k/2}
      + 3^k\iprod{u,\mu-\mu'}^{k}
      \mper
      \Bigr\}
    \end{aligned}
  \end{align}
  Combining the two bounds yields 
  \begin{displaymath}
    \iprod{u,\mu-\mu'}^k \le (2\e)^{k-1} \cdot \
    \Paren{ O(C k)^{k/2} \cdot \iprod{u,(\Sigma+\Sigma')u}^{k/2}
      + 3^k\iprod{u,\mu-\mu'}^k}\mper
  \end{displaymath}
  Since we assume that $\e<0.01$, we can conclude that 
  \begin{equation}
  \label{eq:part-1-mean-estimation}
  \cA_{Y,\e} \cup \cB_{C,k,\ell} \sststile{\ell}{w,x'} \iprod{u,\mu-\mu'}^k\le \delta_1^k \cdot \iprod{u,(\Sigma+\Sigma')u}^{k/2},
  \end{equation} 
  for $\delta_1\le O\Paren{\e^{1-1/k}\cdot \sqrt{C k\,}}$ as desired.

\end{proof}

The following lemma is the sum-of-squares analog of \cref{lem:multiplicative-second-moment}.

\begin{lemma}[Robust identifiability of second moment]
  \label[lemma]{lem:multiplicative-second-moment-sos}
  Let $D$ be the uniform distribution over $X \subseteq \R^d$ with second moment $M \in \R^{d\times d}$.
  Suppose $D$ is $(k,\ell)$-certifiably subgaussian with parameter $C>0$ with $k$ divisible by $4$ and that $Ck \cdot \e^{1-2/k}\le \Omega(1)$.
  Let $Y$ be an $\epsilon$-corruption of $X$. Let $M'$ be the quadratic polynomial $\frac{1}{n} \sum_i\dyad{x'_i}$ in vector-valued variables $x'_1, x'_2, \ldots, x'_n.$ Then, for every $u \in \R^d$,
  \begin{equation}
    \label{eq:multiplicative-second-moment-sos}
    \cA_{Y,\e} \cup \cB_{C,k,\ell}
    \sststile{\ell}{X'} \Set{ \iprod{u,(M-M')u}^{k/2} \leq \delta^{k/2} \iprod{u,M+M')u}^{k/2}  }
    \quad \text{for } \delta\le O\Paren{ Ck \cdot \e^{1-2/k}}
  \end{equation}
  Furthermore, if $\delta\le 0.0001$,
  \[
    \cA_{Y,\e} \cup \cB_{C,k,\ell}
    \sststile{\ell}{X'}
    \Set{
      1-\delta'  \le \tfrac 1 {\iprod{u,Mu}} \iprod{u,M'u} \le  1+\delta'
    } \quad \text{for }\delta'=100\delta\,.
  \]
  Furthermore, the bit complexity of the sum-of-squares proof is polynomial.
\end{lemma}
\begin{proof}
  By an argument analogous to the one in the proof of \cref{lem:mean-estimation-sos}, we can obtain:
  \begin{equation}
    \label{eq:part-1-second-moment-estimation}
    \cA_{Y,\e} \cup \cB_{C,k,\ell}
    \sststile{\ell}{w,x'}
    \Set{
      \iprod{u,(M-M')u}^{k/2}\le \delta^{k/2} \cdot \iprod{u,(M+M'))u}^{k/2}
    }\,,
  \end{equation} 
  for $\delta\le O\Paren{\e^{1-2/k}\cdot C k}$. 
  \cref{lem:sos-fact-1} together with \eqref{eq:part-1-second-moment-estimation} implies for some $\delta'\le 100\delta$,
  \begin{displaymath}
    \cA_{Y,\e} \cup \cB_{C,k,\ell}
    \sststile{\ell}{X'}
    \Set{
      1-\delta' \le \tfrac 1 { \iprod{u,Mu}}\iprod{u,M'(X') u} \leq 1+\delta'
    }\,.
    \qedhere
  \end{displaymath}
\end{proof}

The following lemma is the sum-of-squares analog of \cref{cor:multiplicative-covariance-identifiability}.

\begin{corollary}[Covariance estimation]
  \label[corollary]{cor:multiplicative-covariance-identifiability-sos}
  Let $D$ be the uniform distribution over $X \subseteq \R^d$ with covariance $\Sigma \in \R^{d\times d}$.
  Suppose $D$ is $(k,\ell)$-certifiably subgaussian with parameter $C>0$ with $k$ divisible by $4$ and that $Ck \cdot \e^{1-2/k}\le \Omega(1)$.
  Let $Y$ be an $\epsilon$-corruption of $X$. 
  Let $\Sigma' = \frac{1}{n} \sum_{i = 1}^n \dyad{(x'_i - \mu')}$ be the formal covariance of the uniform distribution over $x'_1,\ldots,x'_n$, where $\mu'=\frac{1}{n} \sum_i x'_i$.

  Then, for every $u \in \R^d$, 
  \begin{equation}
    \label{eq:multiplicative-covariance-sos}
    \cA_{Y,\e} \cup \cB_{C,k,\ell}
    \sststile{\ell}{X'}
    1-\delta   \leq \tfrac 1 {\iprod{u,\Sigma u}}\iprod{u,\Sigma'(X')u} \leq 1+ \delta
    \quad \text{for some } \delta\le O(Ck) \cdot \e^{1-2/k}\mper
  \end{equation}
  Furthermore, the bit complexity of the sum-of-squares proof is polynomial.
\end{corollary}

\begin{proof}
  Let $D_\mu$ be the distributions $D$ shifted by $\mu$ so that $D_\mu$ has expectation $0.$ 
  Then, $D_\mu$ is $(k,\ell)$-certifiably $2C$-subgaussian.
  Further, for $z_i = x'_i-\mu_i$ 
  $\cA_{Y,\e} \sststile{2}{w,X'} \Set{w_i (z_i - (x_i-\mu))=0}.$
  Further, using  \cref{lem:shifts-of-subgaussian-are-subgaussian}, for every $k' \leq k/2$,
  $\cB_{C,k,\ell} \sststile{\ell}{X',u} \frac{1}{n}\sum_i\iprod{x'_i,u}^{2k'} \leq (2Ck')^{k'/2} (\frac{1}{n}\sum_i \iprod{x'_i,u}^{2})^{k/2}.$ 

  Thus, applying \cref{lem:multiplicative-second-moment-sos}, we obtain for every vector $u \in \R^d$,
  \begin{align}
    &\cA_{Y,\e} \cup \cB_{C,k,\ell} \sststile{\ell}{w,X'}\\
    &\iprod{u,(\Sigma-\Sigma'-\dyad{(\mu-\mu')})u}^{k/2} \le \delta_2^{k/2} \cdot \iprod{u,(\Sigma+\Sigma'+\dyad{(\mu-\mu')})u}^{k/2}\\
    &\leq (2\delta_2)^{k/2} \Paren{\iprod{u,(\Sigma+\Sigma')u}^{k/2} + \iprod{u,\dyad{(\mu-\mu')} u }^{k/2}}\mper
  \end{align}
  Here, $\delta_2\le O(Ck\cdot \e^{1-2/k})$.
  Further, 
  \[
    \sststile{k}{X',w} \Set{\iprod{u, (\Sigma-\Sigma') u}^{k/2} \leq 2^{k/2}\iprod{u,(\Sigma-\Sigma'-\dyad{(\mu-\mu')})u}^{k/2} + 2^{k/2} \iprod{u,\dyad{(\mu-\mu')} u}^{k/2}}.
  \]

  By \cref{lem:mean-identifiability} (mean estimation),
  $\iprod{u,\dyad{(\mu-\mu')}u}^{k/2} \le \delta_1^2 \cdot \iprod{u,(\Sigma+\Sigma')u}^{k/2}$, where $\delta_1\le O(Ck\cdot \e^{1-1/k})$.
  Combining the above inequalities, we obtain:
  \begin{equation}
    \cA_{Y,\e} \cup \cB_{C,k,\ell} \sststile{\ell}{w,X'} \iprod{u, (\Sigma-\Sigma') u}^{k/2} \leq O(\delta_2)^{k/2} \iprod{u, (\Sigma + \Sigma') u}^{k/2}. \label{eq:mean-identifiably-eq-1}
  \end{equation}

  Further, 
  \begin{align*}
    \sststile{k}{w,X'} \iprod{u, (\Sigma + \Sigma') u}^{k/2} &= \Paren{\iprod{u,2 \Sigma u} - \iprod{u, (\Sigma-\Sigma')u}}^{k/2}\\
                                                  &\leq 2^{k/2} \iprod{u,2 \Sigma u} + 2^{k/2} \iprod{u, (\Sigma-\Sigma')u}^{k/2}.
  \end{align*}
  Using  this in conjunction with \eqref{eq:mean-identifiably-eq-1}, we obtain:
  \[
    \cA_{Y,\e} \cup \cB_{C,k,\ell} \sststile{\ell}{w,X'}  \iprod{u, (\Sigma-\Sigma') u}^{k/2} \leq O(\delta_2)^{k/2} \iprod{u, \Sigma u}^{k/2}.
  \]

  Now, using \cref{lem:sos-fact-2} with $f = \frac{\iprod{u, (\Sigma-\Sigma')u} }{O(\delta_2)^{k/2} \iprod{u, \Sigma u}}$ completes the proof. 
\end{proof}

The following is the analog of the \cref{cor:mean-identifiability-no-covariance-bound}.

\begin{corollary}[Stronger Mean Estimation]
  \label[corollary]{cor:mean-identifiability-no-covariance-bound-sos}
 Let $D$ be the uniform distribution over $X \subseteq \R^d$ with mean, covariance $\mu \in \R^d$ and $\Sigma \in \R^{d\times d}$.
  Suppose $D$ is $(k,\ell)$-certifiably subgaussian with parameter $C>0$ with $k$ divisible by $4$ and that $Ck \cdot \e^{1-2/k}\le \Omega(1)$. Let $Y$ be an $\epsilon$-corruption of $X$. 
  Let $\mu'(X')$ be the linear polynomial $\frac{1}{n} \sum_i x'_i$ in the vector-valued variables $x'_1, x'_2, \ldots, x'_n.$
  \begin{equation}
    \label{eq:mean-identifiability-final-sos}
    \forall u\in \R^d.
   \cA_{Y,\e} \cup \cB_{C,k,\ell} \sststile{\ell}{w,X'} \pm \Set{ \iprod{u,\mu - \mu'} 
    \le \delta
    \cdot \iprod{u, \Sigma u}^{1/2}\quad \text{ where } \delta= O\Paren{\sqrt{C k\,}\cdot \e^{1-1/k}}}\mper
  \end{equation}
  In particular, $ \cA_{Y,\e} \cup \cB_{C,k,\ell} \sststile{\ell}{w,X'} \norm{\mu -\mu'}\le \delta \norm{\Sigma}^{1/2}$ and $\norm{\Sigma^{-1/2}(\mu-\mu')}\le \delta$.
  Furthermore, the bit complexity of the sum-of-squares proof is polynomial.
\end{corollary}
\begin{proof}
Combining \cref{lem:mean-estimation-sos} and \cref{cor:multiplicative-covariance-identifiability-sos} we have for $\delta_1 = O(\sqrt{Ck} \epsilon^{1-1/k})$ and $\delta_2 = O(Ck) \epsilon^{1-2/k}$,
\begin{align*}
&\cA_{Y,\e} \cup \cB_{C,k,\ell} \sststile{\ell}{w,X'}\\
&\iprod{u,\mu-\mu'(X')}^{k} \le \delta_1^k \iprod{u,(\Sigma + \Sigma') u}^{k/2} \\
&\leq (2\delta_1)^2 \Paren{\iprod{u, \Sigma u}^{k/2} + \iprod{u, \Sigma'u}^{k/2}}\\
&\leq (6\delta_1)^2 \iprod{u, \Sigma u}^{k/2}.
\end{align*}

Observe that $\iprod{u, \Sigma u}^{k/2}$ is a constant, i.e., does not depend on any variable in the SoS proof.
Using \cref{lem:sos-fact-2} with $f = \frac{\iprod{u, \mu-\mu'(X')}}{(6\delta_1)^2 \iprod{u, \Sigma u}^{1/2}}$ now completes the proof. 
\end{proof}

Finally, the following lemma is the sum-of-squares analog of \cref{lem:multiplicative-higher-moment}. The proof is similar to the ones presented before, so we omit it here. 
\begin{lemma}
  \label[lemma]{lem:higher-moment-estimation-sos}
  Let $D$ be the uniform distribution over $X\subseteq \R^d$ with second moments $M_2 \in \R^{d\times d}$.
  Suppose $D$ is $(k,\ell)$-certifiably subgaussian with parameter $C>0$ for even $k$ divisible by $r$.
  Let $Y$ be $\epsilon$-corruption sample of $X$. 
  Let ${M}'_{r} = \tfrac 1n \sum_{i=1}^n x_i'^{\otimes r}$ be a $r$-order tensor with entries that are $r$-degree polynomial in vector-valued variables $x'_1,\ldots,x'_n$.  
  Then, the following inequality has a sum of squares proof of degree $\ell$ in variables of $\cA_{Y,\e} \cup \cB_{C,k,\ell}$ and $u \in \R^d$,
  \begin{equation}
    \cA_{Y,\e} \cup \cB_{C,k,\ell}
    \sststile{\ell}{u,w,X'} \Set{ -\delta \iprod{u,\Paren{M_2 +  {M}'_2}u }^{k/2} \le \iprod{M_{r} - {M}'_{r}, u^{\otimes r}}^{k/r} \le  \delta \iprod{u,\Paren{M_2 + {M}'_2}u}^{k/2} }\,,
  \end{equation}
   where $\delta= O\Paren{C^r k^r}\cdot \e^{1-r/k}$.
  Consequently,
  \begin{equation}
   \cA_{Y,\e} \cup \cB_{C,k,\ell} \cup \Set{ \|u\|_2^2 \leq 1}
   \sststile{\ell}{u,w,X'}
   \Set{
     -\delta \iprod{u,M_2u }^{r/2} \le \iprod{M_{r} - {M}'_{r}, u^{\otimes r}} \le  \delta \iprod{u,M_2u}^{r/2} }\mper
    \end{equation}
  Furthermore, the bit complexity of the sum-of-squares proof is polynomial.
\end{lemma}
\begin{remark}
  We point out a subtle technical difference with respect to the two results above. In \cref{lem:mean-estimation-sos} and \cref{cor:multiplicative-covariance-identifiability-sos}, we do not force $u$ to be a variable in the sum of squares proof system.
  On the other hand, in Lemma \ref{lem:higher-moment-estimation-sos}, the resulting polynomial inequality has a sum of squares proof in $u$ (in addition to the other variables).
  This is important because, the inequalities are degree $2$ (in $u$) for \cref{cor:mean-identifiability-no-covariance-bound-sos} and \cref{cor:multiplicative-covariance-identifiability-sos} while higher degree in general.
  While unconstrained (in $u$) inequalities of degree $\leq 2$ always have a SoS proof, this is not true for higher degrees. The lemma above relies on Proposition \ref{prop:SoS-proof-for-unbounded-variable} to obtain this stronger conclusion. %
\end{remark}

\subsection{Putting things together}
\label{sec:putting-things-together}

We can now complete the proof of Theorem \ref{thm:mean-covariance-main}. 

\begin{proof}[Proof of Theorem \ref{thm:mean-covariance-main}]
We just put together the implications of the above results to analyze Algorithm \ref{alg:moment-estimation-program} here.
Let $S'$ be the true sample used to produce the $\e$-corrupted sample $S$. Then, $\U_{S'}$ - the uniform distribution on $S'$ is $(k,\ell)$-certifiably $2C$-subgaussian by an application of Lemma \ref{lem:sampling-preserves-subgaussianity}.

Let $\zeta$ be the pseudo-distribution of degree at least $\ell$ over variables $w,y$ satisfying $\cA_{Y,\e}\cup \cB_{C,k,\ell}$ output by Algorithm \ref{alg:moment-estimation-program}. 
Define $\mu' = \tfrac 1n \sum_{i=1}^n y_i$ and $ \Sigma' = \tfrac 1n \sum_{i=1}^n (y_i-\hat \mu)(y_i-\hat \mu)^{\top}$ be the polynomials defined in \cref{cor:mean-identifiability-no-covariance-bound-sos} and \cref{cor:multiplicative-covariance-identifiability-sos} respectively. 

We round the pseudo-distribution simply by setting $\hat \mu = \pE_{\zeta} \mu'$. 

From \cref{cor:mean-identifiability-no-covariance-bound-sos}, we have that:
\[
 \cA_{Y,\e} \cup \cB_{C,k,\ell}
    \vdash_{\ell} \Set{ \iprod{u,\mu-\mu'}^2 \le \delta^2 \iprod{u,\Sigma u}}\,,
\]
Thus, for every $u \in \R^d$,
$\pE_{\zeta} \iprod{u,\mu-\mu'}^2 \le \delta^2 \iprod{u,\Sigma u}\,$

By Cauchy-Schwarz inequality for pseudo-distributions, we have:

\begin{equation}
\Paren{\pE_{\zeta} \iprod{u, \mu-\mu'}}^2 \leq  \iprod{u,\mu-\pE_{\zeta}\mu'}^2 \leq \delta^2 \iprod{u,\Sigma u}. \label{eq:main-mean-estimation-sos}
\end{equation}
Using that $\pE \mu' = \hat \mu$ thus yields 
$\Paren{\iprod{u, \mu-\hat \mu}}^2 \leq \delta^2 \iprod{u,\Sigma u}$. 

Using $u \rightarrow \Sigma^{-1/2} u$ thus immediately yields:
\[
\Paren{\iprod{u, \Sigma^{-1/2} (\mu-\hat \mu)}}^2 \leq \delta^2 \|u\|^2.
\]
Taking square roots thus implies $\iprod{u, \Sigma^{-1/2} (\mu-\hat \mu)} \leq \delta\|u\|$ or equivalently, $\|\Sigma^{-1/2}(\mu-\hat \mu)\| \leq \delta$.

On the other hand, taking square roots in \eqref{eq:main-mean-estimation-sos} gives:
$|\iprod{u, \mu-\hat \mu}| = \|\mu - \hat \mu\| \leq \delta \|\Sigma\|^{1/2}$ as required.

Next, we analyze the covariance estimation procedure from Algorithm \ref{alg:moment-estimation-program}.
As before, our rounding is simple: $\hat{\Sigma} = \pE_{\zeta} \Sigma'$. 

Applying \cref{cor:multiplicative-covariance-identifiability-sos} and using that degree $\ell$ pseudo-expectations must respect degree $\ell$ SoS proofs (Fact \ref{fact:sos-soundness}), we get that for every $u \in \R^d$,
whenever $\delta < 1$, $(1-\delta) \Sigma \preceq \hat \Sigma \preceq (1+\delta) \Sigma$ as required.
\end{proof}

The proof of Theorem \ref{thm:higher-moment-main} is entirely analogous. That the inequality in the conclusion has a sum-of-squares proof in $u$ follows from Proposition \ref{prop:SoS-proof-for-unbounded-variable}.

%% file: content/subgaussianity.tex
\section{Certifiably subgaussian distributions}
\label{sec:subgaussianity}

In this section, we give sufficient conditions for distributions to be certifably subgaussian.
As concrete consequences, we derive that (affine transformations of) products of scalar-valued subgaussian distributions are certifiably subgaussian and that certain mixture models are certifiably subgaussian.

For the benefit of the reader, we restate here the definition of certifiable subgaussianity.

\begin{definition*}[Certifiable subgaussianity]
  A distribution $D$ over $\R^d$ with mean $\mu$ is \emph{$(k,\ell)$-certifiably subgaussian} with parameter $C>0$ if there for every positive integer $k'\le k/2$, there exists a degree-$\ell$ sum-of-squares proof of the polynomial inequality
  \begin{equation}
    \label{eq:low-degree-subgaussian-restated}
    \forall u\in \bbS^{d-1}.~
    \E_{D(x)} \iprod{x-\mu,u}^{2k'}
    \le \Paren{C\cdot k' \E_{D(x)} \iprod{x-\mu,u}^2 }^{k'}
  \end{equation}
\end{definition*}

We first note some basic invariance properties of certifiable subgaussianity.

\begin{lemma}[Invariance Under Linear Transformations] \label{lem:linear-transformation-subgaussianity}
Suppose $x \in \R^d$ is a random variable with a $(2k,\ell)$-certifiably $C$-subgaussian distribution. Then, for every matrix $A \in \R^{d \times d}$, the random variable ${Ax}$ is $(2k,\ell)$-certifiably $C$-subgaussian.
\end{lemma}

\begin{proof}
Fix any $v \in \R^d$. Invoking certifiable subgaussianity of $x$, we get that there's a degree $\ell$-SoS proof of the inequality: 
\[
\E[ \langle A x, v \rangle^{2k'} ] = \E[ \langle x, A^{\top} v \rangle^{2k'} ] \leq (Ck' \E[ \langle x, A^{\top} v \rangle^{2} ])^{k'} = (Ck' \E[ \langle Ax, v \rangle^{2} ])^{k'}
\]
\end{proof}

\begin{lemma}[Shifts of certifiable subgaussian distributions]
  \label[lemma]{lem:shifts-of-subgaussian-are-subgaussian}
  \label[lemma]{lem:subgaussian-shift-moment-bound}
  Let $k\in \N$ be even and let $D$ be a mean zero, $(k,\ell)$-certifiably $C$-subgaussian distribution over $\R^d$.
  Then, $D$ satisfies for every number $k' \leq k/2$ and vector $s\in \R^d$,
  \begin{equation}
    \forall u\in \bbS^{d-1}.~\E_{D(x)} \iprod{x+s,u}^{2k'} \le (Ck' \E_{D(x)} \iprod{x+s,u}^2)^{k'}
    \label{eq:shift-subgaussian}
  \end{equation}
  Furthermore, the above inequality \cref{eq:shift-subgaussian} has a degree-$\ell$ sum-of-squares proof.
\end{lemma}

Next, we construct examples of various classes of distributions that are certifiably subgaussian. We start by noting that the standard gaussian distribution is $k$-certifiably $O(1)$-subgaussian for every $k$. 

\begin{lemma}[Certifiable Subgaussianity of Standard Gaussian]
  Let $x$ be the standard gaussian random variable with mean $0$ and covariance $\Id$ on $\R^d$.
  Then, $x$ is $(2k,2k)$certifiably $1$-subgaussian for every $k$.  \label{lem:subgaussianity-of-gaussian}
\end{lemma}

\begin{proof}
By rotational symmetry of the standard gaussian and the fact that $\E[x_i^{2k}] = \Gamma_k$, $\E[ \iprod{x,u}^{2k}] = \Gamma_k \E[\iprod{x,u}^2]^{k}$ for some $\Gamma_k < (2k)^k$. 
\end{proof}

Next, we prove a similar fact about uniform mixtures of arbitrary mean gaussians. 

\begin{lemma} \label{lem:mixtures-of-gaussians-certifiable}
Let $\mu_1, \mu_2, \ldots \mu_q$ be arbitrary.
Let $D$ be the mixture of $\cN(\mu_i,I)$ with mixture weights $\frac{1}{q}$ for each component.
Then, $D$ is $(k,k)$-certifiably $C$-subgaussian with $C=O(q)$.
\end{lemma}
\begin{proof}
We have using Lemma \ref{lem:shifts-of-subgaussian-are-subgaussian} and certifiable $C$-subgaussianity of the gaussian distribution for $C = O(1)$,
\begin{align*}
\E \iprod{x,u}^{2k} &= \frac{1}{q} \sum_i \E_{x \sim \cN(\mu_i,I)} \iprod{x, u}^{2k}\\
&\leq (2Ck)^k\frac{1}{q} \sum_i \Paren{\E_{x \sim \cN(\mu_i,I)} \iprod{x,u}^{2}}^{k}\\ 
&= (2Ck)^k \frac{1}{q} (\iprod{u,\mu_i}^2 + 1)^{k} \leq q^{k-1} (Ck)^k \Paren{\frac{1}{q} \sum_i (\iprod{u,\mu_i}^2+1)}^k.
\end{align*}
On the other hand, $\E \iprod{x,u}^2 = \frac{1}{q} \sum_i (\iprod{u,\mu_i}^2+1)$. Thus, we obtain that $D$ is $(k,k)$-certifiably $(2Cq)$ subgaussian. 
\end{proof}

\begin{lemma}[Sampling Preserves Certifiable Subgaussianity] \label{lem:sampling-preserves-subgaussianity}
Let $\cD$ be a $(k,\ell)$-certifiably subgaussian distribution on $\R^d$. Let $S$ be an i.i.d. sample from $\cD$ of size $n \geq n_0 = \Omega((d\log{(d/\delta)})^{k/2}/\epsilon^2)$. Then, with probability at least $1-\delta$ over the draw of $S$ according to $\cD$, the uniform distribution on $S$, $\U_S$, is $(k,\ell)$ certifiably $(C+\epsilon)$-subgaussian.
\end{lemma}

We will use the following matrix concentration inequality in the proof. 

\begin{fact}[Matrix Rosenthal Inequality] \label{fact:matrix-rosenthal}
Fix $p \geq 1.5$. For a finite sequence $P_k$ for $k \geq 1$ of independent, random psd matrices that satisfy $\E \|P_k\|_{2p}^{2p} < \infty$, 
\[
\Paren{ \E \|\sum_k P_k\|_{2p}^{2p} }^{1/2p} \leq \Paren{\|\sum_k \E P_k\|_{2p}^{1/2} + \sqrt{4p-2} \Paren{\sum_k \|P_k\|_{2p}^{2p}}^{1/4p} }^2\mper
\]
\end{fact}
\begin{proof}
Without loss of generality, assume that $M_2 = \E_D xx^{\top}$ is full-rank. If not, we just work in the range-space of $\E_D xx^{\top}$ instead, as every sample $x$ would live in this subspace.  

Now, consider the affine transformation $x \rightarrow M_2^{-1/2} x$ which makes second moment is now $I$. By Lemma \ref{lem:linear-transformation-subgaussianity}, it is enough to show certifiable subgaussianity of the distribution after this transformation. Thus, without loss of generality, we  assume that the $D$ is isotropic, that is, $M_2 = I.$ 

Now, by sum-of-squares version of the Cauchy-Schwarz inequality
\begin{align} 
\sststile{2k'}{u} \pm (Ck')^{k'/2}(\E_{\U_S} \iprod{x,u}^2)^{k'}-\E_{\U_S} \iprod{x,u}^{2k'} &= \pm \Iprod{(Ck')^{k'/2}\Paren{\E_{\U_S} x^{\otimes 2}}^{\otimes k'} -\E_{\U_S} x^{\otimes 2k'}, u^{\otimes 2k'}} \notag \\
&\leq \|(Ck')^{k'/2}\Paren{\E_{\U_S} x^{\otimes 2}}^{\otimes k'} -\E_{\U_S} x^{\otimes 2k'}\|_{\infty} \|u\|_2^{2k'}\, \label{eq:sos-to-spectral}
\end{align}
where $\|(Ck')^{k'/2}\Paren{\E_{\U_S} x^{\otimes 2}}^{\otimes k'} -\E_{\U_S} x^{\otimes 2k'}\|_{\infty}$ is the spectral norm of the canonical $d^{k'} \times d^{k'}$ flattening of the $2k'$-order tensor $T(\U_S) = (Ck')^{k'/2}\Paren{\E_{\U_S} x^{\otimes 2}}^{\otimes k'} -\E_{\U_S} x^{\otimes 2k'}$. Let $T(D)$ be the corresponding tensor for $D$.

By standard matrix concentration, whenever $n_0 \gg (d \log{d}/\epsilon^2)$, $\E_{\U_S} \iprod{x,u}^2 > 1-\epsilon.$
By standard argument combining Fact \ref{fact:matrix-rosenthal} with Matrix-Chebyshev inequality, we can obtain that with probability at least $1-\delta$ over the draw of the sample $S$ of size $n \geq n_0 = \Omega((d\log{(d/\delta)})^{k/2}/\epsilon^2)$, for any $k' \leq k/2$,
\begin{equation}
\|T(D) - T(\U_S) \|_{\infty} \leq \epsilon. \label{eq:spectral-conc-tensor}
\end{equation}

Thus, along with \eqref{eq:sos-to-spectral}, we obtain that: 
\[
\sststile{2k'}{u} (Ck')^{k'/2}(\E_{\U_S} \iprod{x,u}^2)^{k'}-\E_{\U_S} \iprod{x,u}^{2k'} \leq \epsilon \|u\|_2^{2k'} < 2\epsilon \Paren{\E_{\U_S} \iprod{x,u}^2}^{k'}\mper
\]
Rearranging thus yields that: 
\[
\sststile{2k'}{u} \Paren{(Ck')^{k'/2} +\epsilon)} (\E_{\U_S} \iprod{x,u}^2)^{k'}-\E_{\U_S} \iprod{x,u}^{2k'} \geq 0\,
\]
completing the proof.
\end{proof}

\begin{lemma}[Certifiable Subgaussianity of Mixtures]
Let $D_1$ and $D_2$ be two $(k,\ell)$-certifiably $C$-subgaussian distributions with covariances $\Sigma_1, \Sigma_2$ respectively. Suppose further that $ \Delta \Sigma_1 \succeq  \Sigma_2 \succeq \Sigma_1$.

Let $D$ be a mixture of $D_1$ and $D_2$ with mixture weights $1-\lambda$ and $\lambda$ respectively such that $\lambda < \Delta^{-k/2}$. 

Then, $D$ is $(k,\ell)$-certifiably $2C$-subgaussian. \label{lem:certifiably-subgaussian-mixtures}
\end{lemma}

\begin{proof}
Assume that $D_1$ and $D_2$ are mean zero. 
We have for any $u$,

$\vdash_{2} \E_{D_1} \iprod{x,u}^2 \leq \E_{D} \iprod{x,u}^2 \leq (1-\lambda + \lambda \Delta) \E_{D_1} \iprod{x,u}^2$.

On the other hand, for any $k' \leq k/2$, $\vdash_{2k'} \E_{D} \iprod{x,u}^{2k'} \leq (1-\lambda) \Paren{Ck\E_{D_1} \iprod{x,u}^2}^{k'} + \lambda \E_{D_2} \Paren{Ck'\E_{D_2} \iprod{x,u}^2}^{k'} \leq (1-\lambda + \lambda \Delta^{k'}) \Paren{Ck'\E_{D_1} \iprod{x,u}^2}^{k'}$. 

Thus, $D$ is $(k,\ell)$-certifiably $2C$-subgaussian. 
\end{proof}

The proof of this lemma follows immediately from the following proposition. 

\begin{proposition}[Invariance Under Additive Shifts]
Let $k\in \N$ be even and let  $X$ be a mean zero, scalar random variable satisfying $(\E X^k)^{1/k}\le  (C \E X^2)^{1/2}$. Then, every $s\in R$ satisfies $(\E (X+s)^k)^{1/k}\le 2(C\E (X+s)^2)^{1/2}$.
\end{proposition}

\begin{proof}
  By the triangle inequality for $L_k$ norms, we have $\Paren{\E (X+s)^k}^{1/k}\le \Paren{\E X^k}^{1/k} + \abs{s}$.
  Therefore, $\Paren{\E (X+s)^k}^{1/k}\le \Paren{C \E X^2}^{1/2} + \abs{s}$.
  On the other hand $\E (X+s)^2 = \E X^2 + s^2$ using the fact that $\E X = 0$.
  Thus, $2(\E (X+s)^2)^{1/2} \geq (\E X^2)^{1/2} + \abs{s}$.
  Combining the last two inequalties yields the lemma.
\end{proof}

\begin{lemma}[Certifiable Subgaussianity of Product Distributions]
Let $\cD = \cD_1 \otimes \cD_2 \otimes \cdots \otimes \cD_d$ be a product distribution on $\R^d$ such that each $\cD_i$ has ean zero, variance $1$ and certifiably subgaussian with constant $C_i$. Then, $\cD$ is $(2k,2k)$-certifiably subgaussian with constant $C = \max_i C_i$. \label{lem:product-subgaussianity}
\end{lemma}

\begin{proof}

Let $1 \leq C = \max_i C_i$. 
Observe that in that case, $\E_{\cD_i} x^{2j} \leq C^{2j} \cdot \E_{\cN(0,1)} x^{2j}$ for every $j \leq 2k$.

Let $\cD'$ be the distribution obtained by the following process: 1) sample $x \sim \cD$ and 2) output $x$ and $-x$ with uniform probability. By linearity, $\E_{\cD} \iprod{x,u}^{2k} =  \E_{\cD'} \iprod{x,u}^{2k}$.

Call a multi-set $S$ of $[d]$ even if every individual element in $S$ appears even number of times. 

Let $S$ be an even multi-set of size $|S| = 2k$. Then, $\E_{\cD'} x^S = \E_{\cD} x^S \leq C^{|S|} \cdot \E_{\cN(0,I)} x^S$ using the sub-gaussianity of each $x_i$.

Now, let $u \in \R^d$. Then, $\E_{\cD} \iprod{x,u}^2 = \E_{\cD'} \iprod{x,u}^2 = \|u\|_2^2$.
We have:

\[
\E_{\cD'} \langle x, u \rangle^{2k} = \sum_{S: \text{ even multi-set } }  \E[ x^S ] \cdot  u^S.
\]
Since $S$ is even, we can upper bound the term corresponding to $S$ above by $C^{|S|} \E_{\cN(0,I)} x^S u^S$. Thus,
$\E_{\cD'} \langle x, u \rangle^{2k} \leq C^{|S|} \sum_{S: \text{ even multi-set } } \E_{\cN(0,I)} x^S u^S$. Thus, the RHS is at most $\sigma^{2k} \cdot \E_{x \sim \cN(0,I)} \iprod{x,u}^{2k}$. 

By rotational symmetry,
$\E_{x \sim \cN(0,I)} \iprod{x,u}^{2k} = \Gamma_k \cdot (\E_{x \sim \cN(0,I)} \iprod{x,u}^{2})^k = \Gamma_k \cdot (\E_{\cD'} \iprod{x,u}^{2})^k$, where $\Gamma_k\le (2k)^k$.

Thus, $\E_{\cD'} \langle x, u \rangle^{2k} \leq (2k)^k C^{2k} (\E_{\cD'} \iprod{x,u}^{2})^k$.
And thus, $\E_{\cD} \langle x, u \rangle^{2k} \leq (2k)^k C^{2k} (\E_{\cD} \iprod{x,u}^{2})^k$. Showing that $\cD$ is $(2k,2k)$ certifiably $C$-subgaussian.
\end{proof}

%% file: content/application-ica.tex
\section{Applications}
\subsection{Outlier-Robust Independent Component Analysis}
\label{sec:outl-robust-indep}

In this section, we apply our robust moment estimation algorithm to give an efficient algorithm for outlier-robust independent component analysis. We recall the definition of subgaussian random variables first. 

\begin{definition}[Subgaussian Distributions]
A distribution $D$ on $\R$ is said to be $K$-subgaussian if for every $p \geq 1$, $\E_{x \sim D} [x^p]^{1/p} \leq K (\E_{z \sim \cN(0,1)}[z^p])^{1/p} \leq K\cdot \sqrt{p}$. A distribution $D$ on $\R^n$ is said to be $K$-subgaussian if for every unit vector $v \in \R^n$, $\langle x, v \rangle$ is $K$-subgaussian random variable on $\R$.
\end{definition}

Next, we describe the ICA problem:

\begin{definition}[Fully Determined Independent Component Analysis]
Let $A$ be an unknown \emph{mixing matrix} in $\R^{d \times d}$ of full rank. 
Let $x$ be an unknown product distribution with $C$-subgaussian marginals satisfying $\E[x_i^2] = 1$ and $\E[x_i^4]-1 \neq 0$ for every $i$.
Given samples from $y = \{Ax\}$ for i.i.d. samples $x$, estimate a matrix $\hat{A}$ such that there's a permutation $\pi$ and some signing $\sigma \in \on^{d}$ on columns of $\hat{A}$ satisfying $\|A_i - \sigma_i \cdot \hat{A}_{\pi(i)}\|_2 \leq \delta \|A_i\|$ for the smallest possible $\delta$. 
\end{definition}
\begin{remark}
\begin{enumerate}
	\item The assumption that $\E[x_i^2] =1$ is without the loss of generality. Since we only see samples from $\{Ax\}$ one can't uniquely recover the scalings on the column $A_i$ and $\E[x_i^2]$ simultaneously for any $i$.
	\item The goal is to recover a matrix $A$ whose columns are close to that of $A$ up to permutation and signs. This is again, necessary, since the order of the columns cannot be uniquely recovered from samples of $\{Ax\}$.
	\item Information theoretically, columns of $A$ can be recovered (up to permutation and scaling) if  at most one $x_i$ is has gaussian distribution.
\end{enumerate}
\end{remark}

Our main result is a outlier-robust algorithm for ICA. Our algorithm is based on 4th order tensor decomposition.
We will use noise-tolerant generalization of the FOOBI algorithm due to Cardoso \cite{MR2473563-DeLathauwer07} presented in \cite{DBLP:conf/focs/MaSS16}.

\begin{theorem}[Outlier-Robust ICA]
For every $C\ge 1$ and even $k\in \N$, there exists a $(dn)^{O(k)}$ that given a (corrupted) sample $S\subseteq \R^d$ of size $n$ outputs component estimates $\hat a_1,\ldots,\hat a_d\in \R^d$ with the following guarantees:
Suppose $A \in \R^{d \times d}$ is a non-singular matrix with condition number $\kappa$ and columns $a_1,\ldots,a_d\in \R^d$.
Suppose  $\mathbf x$ is a centered random vector with $d$ independent coordinates such that every coordinate $i\in [d]$ satisfies $\E[\mathbf x_i^2] = 1$, $\E[\mathbf x_i^4] -3 = \gamma \neq 0$, and  $\E[\mathbf x_i^{k}]^{1/k} \leq \sqrt{Ck}$.
Then, if $S$ is an $\e$-corrupted sample of size $\card{S}\ge n_0$ from the distribution $\set{A \mathbf x}$, where $n_0\le (C+\kappa +d)^{O(k)}$, the component estimates satisfy with high probability
  \begin{equation}
    \max_{\pi \in S_d}\min_{i\in [d]} \iprod{A^{-1}\hat a_i, A^{-1} a_{\pi(i)}}^2 \ge 1-\delta
    \quad \text{ for } \delta < (1+\tfrac 1{\abs{\gamma}}) \cdot O(C^2 k^2) \cdot \epsilon^{1-4/k}
    \mper
    \label{eq:ica-guarantee}
  \end{equation}
\end{theorem}

Our algorithm is simple. It first estimates the 2nd and 4th moments of the observed distribution from an $\epsilon$-corrupted sample using Algorithm \ref{alg:moment-estimation-program} and then applies a blackbox tensor decomposition algorithm to a ``whitened'' version of 4th moments. 

\paragraph{Tensor Decomposition} 
We state the tensor decomposition algorithm (implicit in \cite{DBLP:conf/focs/MaSS16}, analog of Theorem 1.1 for 4th order tensor decomposition) we use before proceeding. 

It is convenient to define the following relaxation of the injective tensor norm in describing the guarantees of this algorithm.

\begin{definition}[Sum of Squares Norms]
Given a tensor $T \in (\R^d)^{\otimes 2t}$ for any $2t$, the degree $2t$ \emph{sum of squares norm} of $T$, $\|T\|_{\sos_{2t}}$ is defined by 
\[ 
\Paren{\|T\|_{\sos_{2t}}}^{2t} = \sup_{\mu(u): \text{ deg $2t$ pseudo-dist over sphere }} \Paren{\pE_{\mu(u)} \iprod{T, u^{\otimes 2t}}}. \label{eq:def-sos-norm}
\]
\end{definition}

\begin{fact}[\cite{DBLP:conf/focs/MaSS16}] \label{fact:tensor-decomposition-blackbox}
There exists a polynomial time algorithm that given a symmetric $4$-tensor $T \in \Paren{\R^d}^{\otimes 4}$, outputs a set of vectors $\{c_1', c_2', \ldots, c_n'\} \subseteq \R^d$, such that for every orthonormal set $\{c_1, c_2, \ldots, c_n\} \in \R^d$, there's a permutation $\pi:[n] \rightarrow [n]$ satisfying 
\[
\max_{i} \|c_i - c_{\pi(i)}'\|^2 \leq O(1) \cdot \| \frac{T}{\|T\|_{\sos_4}} -\sum_{i =1}^n c_i^{\otimes 4}\|_{\sos_4}.
\]
\end{fact}

Without loss of generality, we can assume that the distribution $x$ satisfies $\E[x_i^{2j-1}] = 0$ for every $j > 0$. This can be accomplished by negating each sample independently with probability $1/2$. The fraction of corruptions in the modified sample is still $\epsilon$. 

\begin{mdframed}
  \begin{algorithm}[Algorithm for Outlier-Robust ICA via Sum of Squares]
    \label[algorithm]{alg:outlier-robust-ica}\mbox{}
    \begin{description}
    \item[Given:]
      $\e$-corrupted sample from $\{Ax\}$: $Y=\set{y_1,\ldots,y_n}\subseteq \R^d$. %
    \item[Estimate:]
      Mixing Matrix $A \in \R^{d \times d}$.    
    \item[Operation:]\mbox{}
      \begin{enumerate}
      \item 
      	Apply Algorithm \ref{alg:moment-estimation-program} to estimate the covariance $\hat{\Sigma}$ and 4th moments $\hat{M}_4$ of sample $Y$.
      \item Define the ``whitened'' estimated 4th moment tensor $\hat{M}_4'$ so that for every $u \in \R^d$, $\iprod{\hat{M}_4',u^{\otimes 4}} = \iprod{\hat{M}_4, (\Sigma^{-1/2} u)^{\otimes 4}}$.

       Let $\sigma_1, \sigma_2, \ldots, \sigma_d$ be the columns of $\hat{\Sigma}^{-1/2}$. Then, the above is equivalent to setting: $\hat{M}_4' = \sum_{j_1, j_2, \ldots, j_4} \sigma_{j_1} \otimes \sigma_{j_2} \otimes \sigma_{j_3} \otimes \sigma_{j_4} \cdot (\hat{M}_4)(j_1, j_2, j_3, j_4)$. 
       \item Let $\hat{T} = \hat{M}_4'-3I \otimes I$.
      \item 
        Apply Tensor Decomposition from Theorem \ref{fact:tensor-decomposition-blackbox} to $\hat{T}$ to recover components $\hat{a}_1, \hat{a}_2,\ldots, \hat{a}_d$.
        \item Output $\hat{\Sigma}^{1/2} \hat{A}$ where $\hat{A} = (\hat{a}_1, \ldots, \hat{a}_d)$.
      \end{enumerate}
    \end{description}    
  \end{algorithm}
\end{mdframed}

We now analyze the algorithm. The following lemma shows that the components of the tensor $T$ that we decompose are close to the columns of $A$.  
\begin{lemma}
Let $T = \hat{M}_4'-3I \otimes I$ for the whitenened 4th moments $\hat{M}_4'$.
Let $a_1, a_2, \ldots, a_d$ be the columns of $A$ and let $b_i = (AA^{\top})^{-1/2} a_i$ be the whitened versions. Then,

\[
\| \hat{T} - \gamma \sum_i b_i^{\otimes 4}\|_{\sos_4} \leq \delta_4^{1/4}(2+ 3\delta_2)^{1/4}.
\]
\end{lemma}

\begin{proof}

First, let's understand why the tensor decomposition should produce the columns of $A$ if all the moment estimates were exactly correct. We will then account for the estimation errors. 

For the distribution $\{Ax\}$, we can compute $M_2 = \Sigma = \E[ (Ax) (Ax)^{\top}] = AA^{\top}$ using $\E[xx^{\top}] = I$. Thus, for $b_i = \Sigma^{-1/2} a_i$, $\sum_i b_i b_i^{\top} = \sum_i \Sigma^{-1/2} a_i a_i^{\top} \Sigma^{-1/2} = I$. 

Similarly, we have 
\begin{align*}
\iprod{M_4,u^{\otimes 4}} = \E[ \iprod{Ax, u}^4] &= \sum_i \E[x_i^4] \iprod{a_i,u}^4 + 3\sum_{i \neq j} \E[x^2_i x^2_j] \iprod{a_i,u}^2 \iprod{a_j, u}^2\\
&= (\gamma-3) \iprod{a_i,u}^4 + 3(\sum_i \iprod{a_i,u}^2)^2. \end{align*}

The whitened and shifted 4th moment $T$ then, by definition, satisfies
\begin{align*}
\iprod{T, u^{\otimes 4}} = \iprod{M_4'-3I\otimes I,u^{\otimes 4}}  &= \iprod{M_4, (\Sigma^{-1/2}u)^{\otimes 4}}-3\|u\|_2^4\\
&= \gamma \iprod{b_i,u}^4 + 3(\sum_i \iprod{b_i,u}^2)^2 - 3\|u\|_2^4 = \gamma \iprod{b_i,u}^4.
\end{align*}

As a result, $\vdash_{u,4} \iprod{T - \gamma \sum_i b_i^{\otimes 4},u^{\otimes 4}} = 0$ and thus, $\|T  - \gamma \sum_i b_i^{\otimes 4}\|_{\sos_4} = 0$. 

Now, let $\hat{\Sigma}, \hat{M}_4, \hat{M}_4', \hat{T}, \hat{b}_i$ stand for their estimated (via Algorithm \ref{alg:moment-estimation-program}) counterparts. To complete the analysis, we now understand the error $\|T-\hat{T}\|_{\sos_4}$. The proof is then complete by using:
\[
\|\hat{T}- \gamma \sum_i b_i^{\otimes 4} \|_{\sos_4} \leq \|\hat{T}- T\|_{\sos_4} + \|T- \gamma \sum_i b_i^{\otimes 4}\|_{\sos_4} = \|\hat{T}- T\|_{\sos_4}.
\]

First, we check that our algorithm (Algorithm \ref{alg:moment-estimation-program}) can indeed be applied to this setting. 
Since $x$ is a product distribution with $i$th marginal $C_i$-subgaussian. By Lemma \ref{lem:product-subgaussianity}, $x$ is $(k,k)$-certifiably $C$-subgaussian. By Lemma \ref{lem:linear-transformation-subgaussianity}, $\{Ax\}$ is $(k,k)$-certifiably subgaussian. \footnote{Negating samples with probability $1/2$ preserves 2nd and 4th moments and doesn't affect certifiable subgaussianity.} Finally, applying Lemma \ref{lem:sampling-preserves-subgaussianity} shows that with probability at least $1-1/n^2$, the sample the facts above all hold for sample $Y$ of size $m \geq \tilde{\Omega}(d^{k/2}/\epsilon^2)$.

We will use $\Sigma, M_4, M_4', T, b_i = \Sigma^{-1/2} a_i$ for the covariance, 4th moment, whitenened 4th moment, shifted, whitenened 4th moment and orthogonalized columns of $A$ respectively. 

By \cref{cor:multiplicative-covariance-identifiability-sos}, for every $u$, $\iprod{u,\hat{\Sigma} u} = (1\pm \delta_2) \iprod{u,\Sigma u}$ for some $\delta_2 < O(Ck) \epsilon^{1-2/k}$. Thus:

\begin{equation}
(1-\delta_2)I\preceq \sum_i \hat{b}_i \hat{b}_i^{\top} \preceq (1+\delta_2) I. \label{eq:whitened-eigenvalue}
\end{equation} 
By Theorem \ref{thm:higher-moment-main}, for $\delta_4 = O(C^2k^2) \epsilon^{1-4/k}$, 
\begin{equation}
\sststile{4}{u,k} \pm \Paren{\iprod{M_4 - \hat{M}_4, u^{\otimes 4}}} \le \delta_4 \cdot \Paren{\iprod{\Sigma, u^{\otimes 2}}^2 + \iprod{\hat{\Sigma}, u^{\otimes 2}}^2}. \label{eq:4th-moment-approx}
\end{equation}
Thus,
\begin{align*}
\sststile{4}{u,k} \pm \Paren{\iprod{M_4 - \hat{M}_4, (\hat{M}_2^{-1/2}u)^{\otimes 4}}} &\le \delta_4 \cdot \Paren{\iprod{M_2, (\hat{M}_2^{-1/2}u)^{\otimes 2}}^2 + \iprod{\hat{M}_2, (\hat{M}_2^{-1/2}u)^{\otimes 2}}^2}\\
&\leq \delta_4(2+3\delta_2)\|u\|_2^4
\end{align*}
Or,
\begin{align*}
\sststile{4}{u,k} \pm \Paren{\iprod{M_4' - \hat{M}_4', u^{\otimes 4}}} &\le \delta_4 (2+3\delta_2)\|u\|_2^4\\
\end{align*}
Since $M_4' - \hat{M}_4' = T-\hat{T}$, this gives:
\begin{align}
\sststile{4}{u,k} \pm \Paren{\iprod{T-T', u^{\otimes 4}}} \le \delta_4(2+3\delta_2)\|u\|_2^4 \label{lem:error-in-whitened-tensor}
\end{align}

\end{proof}

The proof of Theorem \ref{thm:outlier-robust-ica} then follows easily. 

\begin{proof}[Proof of Theorem \ref{thm:outlier-robust-ica}]
We are now ready to apply Fact \ref{fact:tensor-decomposition-blackbox} to obtain that decomposing the tensor $T$ recovers components $c_1, c_2, \ldots, c_d$ such that there's a permutation $\pi:[d] \rightarrow [d]$ such that $\|c_i - b_{\pi(i)}\|_2^2 \leq O(1) \cdot \frac{1}{|\gamma|} \cdot \delta_4(2+ 3 \delta_2)$.

Applying the reverse whitening transform $\hat{a}_i = \hat{\Sigma}^{1/2} c_i$ then recovers the set of vectors $a_1, a_2, \ldots, a_d$ satisfying the requirements of Theorem \ref{thm:outlier-robust-ica}.
\end{proof}

%% file: content/mixturesofgaussians.tex
\subsection{Outlier-Robust Mixtures of Gaussians}
\label{subsec:outl-robust-mixtures-gaussians}

In this section, we describe our outlier-robust algorithm for learning mixtures of spherical gaussians.  As before, this algorithm can be seen as a direct analog of the non-robust algorithm that learns mixtures of spherical gaussians with linearly independent means by decomposing the 3rd order moment tensor of the observed sample.   %

\begin{theorem}[Outlier-Robust Mixtures of Gaussians] \label{thm:mixture-of-gaussian-technical}
Let $D$ be mixtures of $\cN(\mu_i,I)$ for $i \leq q$ with uniform\footnote{While our algorithm generalizes naturally to arbitrary mixture weights, we restrict to this situation for simplicity} mixture weights. Assume that $\mu_i$s are linearly independent and, further, let $\kappa$ be the smallest non-zero eigenvalue of $\frac{1}{q}\sum_i \mu_i \mu_i^{\top}$.

Given an $\epsilon$-corrupted sample of size $n \geq n_0= \Omega((\kappa d\log{(d)})^{k/2})$, for every $k \geq 4$, there's a $\poly(n) d^{O(k)}$ time algorithm that recovers $\hat{\mu}_1, \hat{\mu}_2,\ldots \hat{\mu}_q$ so that there's a permutation $\pi:[q] \rightarrow [q]$ satisfying \[
\max_i \| (\frac{1}{q}\sum_i \mu_i \mu_i^{\top})^{-1/2}(\hat{\mu}_i - \mu_{\pi(i)})\| \leq O(\sqrt{qk}) \epsilon^{1-1/k}/\sqrt{\kappa} + O(qk) \epsilon^{3/2-2/k}/\kappa +  O(\sqrt{qk}) \epsilon^{1/3-1/k}.
\]
\end{theorem}

We will rely on 3rd order tensor decomposition algorithm (Theorem 1.1) from \cite{DBLP:journals/corr/MaSS16} here.

\begin{fact}[\cite{DBLP:conf/focs/MaSS16}] \label{fact:3-tensor-decomposition-blackbox}
There exists a polynomial time algorithm that given a symmetric $3$-tensor $T \in \Paren{\R^d}^{\otimes 3}$, outputs a set of vectors $\{c_1', c_2', \ldots, c_n'\} \subseteq \R^d$, such that for every orthonormal set $\{c_1, c_2, \ldots, c_n\} \in \R^d$, there's a permutation $\pi:[n] \rightarrow [n]$ satisfying 
\[
\max_{i} \|c_i - c_{\pi(i)}'\|^2 \leq O(1) \cdot \| \frac{T}{\|T\|_{\sos_4}} -\sum_{i =1}^n c_i^{\otimes 3}\|_{\sos_4}.
\]
\end{fact}

\begin{mdframed}
  \begin{algorithm}[Algorithm for Outlier-Robust Mixture of Spherical Gaussians via Sum of Squares]
    \label[algorithm]{alg:outlier-robust-mog}\mbox{}
    \begin{description}
    \item[Given:]
      $\e$-corrupted sample $X$ from $\sum_i \frac{1}{q} \cN(\mu_i,I)$ for $\mu_1, \mu_2, \ldots, \mu_q$. %
    \item[Estimate:]
      $\mu_1, \mu_2, \ldots, \mu_q$. 
    \item[Operation:]\mbox{}
      \begin{enumerate}
      \item 
      	Apply Algorithm \ref{alg:moment-estimation-program} to estimate the first three moments, $\hat{M}_1, \hat{M}_2,$ and $\hat{M}_3$ of sample $X_0$.
      \item Define the ``whitened'' version of the estimated 3rd moment tensor $\hat{T}$ so that for every $u \in \R^d$, $\iprod{\hat{T},u^{\otimes 3}} = \iprod{\hat{M}_3-3\hat{M}_1 \otimes I, ((\hat{M}_2-I)^{-1/2} u)^{\otimes 3}}$.

      Concretely, set $\hat{M}_1' = (\hat{M}_2-I)^{-1/2} \hat{M}_1$. Let $\sigma_1, \sigma_2, \ldots, \sigma_d$ be the columns of $(\hat{M}_2-I)^{-1/2}$. Set $\hat{T} = \sum_{j_1, j_2, j_3} \sigma_{j_1} \otimes \sigma_{j_2} \otimes \sigma_{j_3} \cdot (\hat{M}_3- 3\hat{M_1}' \otimes I) (j_1, j_2, j_3)$.
      \item 
       Apply Tensor Decomposition from Theorem \ref{fact:tensor-decomposition-blackbox} to $\hat{T}$ to recover components $\hat{\nu}_1, \hat{\nu}_2,\ldots, \hat{\nu}_q$.
        \item Output $\hat{\mu}_i = (\hat{M}_2-I)^{1/2} \hat {\nu}_i$ for $i \leq q$ as the estimated means.%
      \end{enumerate}
    \end{description}    
  \end{algorithm}
\end{mdframed}

The following is the main technical lemma required in the proof of \ref{thm:mixture-of-gaussian-technical}. 

\begin{lemma}
Let $Y$ be an $\epsilon$-corrupted sample of size $n > n_0 = \Omega( (\kappa d \log d)^{k}/\epsilon^2)$ from $D_0 = \sum_i \frac{1}{q} \cN(\mu_i,I)$. 
Let $\hat{M}_i$ be the estimated ith moment of $D_0$ given by Algorithm \ref{alg:moment-estimation-program} when run on $Y$.
Let $\hat{M}_i'$ be the whitening of the estimated $i$th moment of $D_0$ as in Algorithm \ref{alg:outlier-robust-mog}.

Then, with high probability over the draw of the $\epsilon$-corrupted sample $Y$, for $\hat{T} = \hat{M}'_3-3 \hat{M}_1' \otimes I$ and $\nu_i = (M_2-I)^{-1/2} \mu_i$ for $i \leq q$ and $\delta_i = O(qk)^{i/2} \epsilon^{1-i/k}$ for $i = 1,2,3$, 
\[
\|\hat{T}- \frac{1}{q} \sum_i \nu_i^{\otimes 3} \|_{\sos_4} \leq \delta_1/\sqrt{\kappa} + \delta_1\delta_2^{1/2}/\kappa + 2\delta_3^{1/3}.
\] \label{lem:tensor-is-useful-mog}
\end{lemma}

\begin{proof}
As before, we will first analyze the tensor $T$ assuming all estimates of the moments are exactly correct and then account for the estimation errors. We will write $M_1, M_2, M_3$ for the true first three moments of the input mixture of gaussian and $\Sigma = M_2-I$. We will let $\hat{M}_1, \hat{M}_2, \hat{M}_3, \hat{\Sigma}$ be the corresponding quantities estimated  by our algorithm.  We will write $M_1', M_3'$ for the true whitened moments and $\hat{M}_1', \hat{M}_3'$, the estimated counterparts.  

By direct computtion, $\iprod{M_3, u^{\otimes 3}} = \frac{1}{q} \sum_i (\iprod{\mu_i,u}^3 + 3\iprod{\mu_i,u})$. Thus, $\iprod{M_3-3M_1 \otimes I, u^{\otimes 3}} = \frac{1}{q} \sum_i \iprod{\mu_i,u}^3$.

Now, by definition, $\iprod{M_3' - 3M_1' \otimes I, u^{\otimes 3}} = \iprod{M_3-3M_1 \otimes I, (\Sigma^{-1/2}u)^{\otimes 3}}$.

As a result, for $T = M_3'-3M_1' \otimes I$, we have: $\|T - \frac{1}{q}\sum_i (\Sigma^{-1/2}\mu_i)^{\otimes 3}\|_{\sos_4} = 0$.

We now account for the estimation errors and bound $\|T - \hat{T}\|_{\sos_4}$.
First, we use Lemma \ref{lem:mixtures-of-gaussians-certifiable} and Lemma \ref{lem:sampling-preserves-subgaussianity} to conclude that with high probability, the uniform distribution on the input sample is $k$-certifiably $O(q)$-subgaussian. 

Now (via triangle inequality), 
\begin{equation}
\|T - \hat{T}\|_{\sos_4}\leq \|M_3' - \hat{M}_3'\|_{\sos_4} + 3\|M_1' \otimes I - \hat{M}_1' \otimes I \|_{\sos_4}.
\label{eq:basic-split}
\end{equation}

Now, 
\begin{equation}
|\iprod{M_1'-\hat{M}_1',u}| \leq |\iprod{\Sigma^{-1/2}(M_1-\hat{M}_1),u}| + |\iprod{(\Sigma^{-1/2}-\hat{\Sigma}^{-1/2}) \hat{M}_1,u}| \label{eq:err-mean} 
\end{equation}

From Theorem \ref{thm:mean-covariance-main}, for $\delta_1 = O(\sqrt{qk}) \epsilon^{1-1/k}, \delta_2 = O(qk) \epsilon^{1-2/k}$, $\iprod{M_1-\hat{M}_1, \Sigma^{-1/2} u} \leq \delta_1 \iprod{u, M_2 \Sigma^{-1} u}^{1/2}$. 

Using that $M_2 = I + \frac{1}{q} \sum_i \mu_i \mu_i^{\top}$, $M_2 \Sigma^{-1} = (\frac{1}{q}\sum_i \mu_i \mu_i^{\top})^{-1} + I$. Thus, the first term in \eqref{eq:err-mean} can be upper bounded by $\delta_1 \|u\|_2 + \delta_1 \iprod{u,\Sigma^{-1},u}^{1/2} \leq \delta_1 (1+1/\sqrt{\kappa}) \|u\|_2$.

From Theorem \ref{thm:mean-covariance-main}, for any $u$, $\iprod{u,(M_2-\hat{M}_2)u} \leq \delta_2 \iprod{u,M_2 u}$. Writing $M_2 = \Sigma+I$ for $\Sigma = \frac{1}{q} \mu_i \mu_i^{\top}$, we obtain that $\iprod{u,\hat{\Sigma} u}= (1 \pm O(\delta_2/\kappa)) \iprod{u, \Sigma u}$. %

Thus, $\iprod{(\Sigma^{-1/2}-\hat{\Sigma}^{-1/2}) \hat{M}_1,u} \leq O(\delta_2^{1/2}/\sqrt{\kappa}) \cdot \iprod{\Sigma^{-1/2} \hat{M}_1, u}$. 
Writing $\iprod{\Sigma^{-1/2} \hat{M}_1, u} = \iprod{ \Sigma^{-1/2} (\hat{M}_1 - M_1) u} + \iprod{ \Sigma^{-1/2} M_1 u}$ and using that $\frac{1}{q} \sum_i \iprod{\Sigma^{-1/2} \mu_i, u} \leq \frac{1}{\sqrt{q}}\|u\|_2$, we finally obtain anupper bound on $\eqref{eq:err-mean}$ of: $\delta_1 (1+1/\sqrt{\kappa}) + O(\delta_2^{1/2}/\sqrt{\kappa}) \delta_1 \Paren{ 1+1/\sqrt{\kappa} + \frac{1}{\sqrt{q}}} \|u\|_2 \leq O(\delta_2^{1/2} \delta_1/\kappa + \delta_1/\sqrt{\kappa}) \|u\|_2$.

For the first term of \eqref{eq:basic-split}, we apply Theorem \ref{thm:higher-moment-main} to get that for $\delta_3 = O(qk)^{3/2} \epsilon^{1-3/k}$, 

\[
\sststile{6}{u} \iprod{\hat{M_3} - M_3,u^{\otimes 3}}^2 \leq \delta_3^2 \iprod{M_2,u^{\otimes 2}}^{3} \leq 2\delta_3 \iprod{\Sigma,u^{\otimes 2}}^{3} \|u\|_2^{3}.
\]
 As a result, $\sststile{6}{u} \iprod{\hat{M_3}' - M_3',u^{\otimes 3}}^2 \leq 2\delta_3^2 \|u\|_2^{3}$ giving us an upper bound on the first term in the RHS of \eqref{eq:basic-split}.

This completes the proof.

\end{proof}

We can now complete the proof of Theorem \ref{thm:mixture-of-gaussian-technical}.

\begin{proof}[Proof of Theorem \ref{thm:mixture-of-gaussian-technical}]
Applying Fact \ref{fact:3-tensor-decomposition-blackbox} to $\hat{T}$ as in Lemma \ref{lem:tensor-is-useful-mog} yields that the resulting components satisfy $\min_{\pi:[q]\rightarrow [q]} \max_i \| \nu_i - \hat{\nu}_{\pi(i)}\|_2 \leq O(1) \delta_1/\sqrt{\kappa} + \delta_1\delta_2^{1/2}/\kappa + 2\delta_3^{1/3}$. The theorem now follows. 
\end{proof}

%% file: content/lower-bounds.tex
\section{Information Theoretic Lower Bounds} \label{sec:lower-bounds}
In this section, we give simple examples that show that our robust moment estimation results achieve sharp information theoretic limits for (certifiably) subgaussian distributions. In both results below, for any choice of $k$, we will show that there are two distributions that 1) both satisfy $(k,k)$-certifiable $C$-subgaussianity and 2) are at most $\epsilon$ apart in total variation distance but have low-degree moments that are exactly as far apart as the estimation error in our algorithm. In fact, these distributions are really simple and in fact just one dimensional and thus the certifiability follows from generic resuls about sum of squares representations for univariate non-negative polynomials.

We begin by showing a lower bound on the error incurred in robust estimation of mean.

\begin{lemma}
There exist two distributions $D_1$ and $D_2$ on $\R$ such that:
\begin{enumerate}
\item $D_1$, $D_2$ are both $2$-subgaussian in all moments up to $k$. 
\item $d_{TV}(D_1, D_2) \le \epsilon$.
\item $\|\E_{D_1}x- \E_{D_2} x\| \geq \sqrt{k} \epsilon^{1-1/k}$.
\end{enumerate}
\label{lem:info-theoretic-lowerbound-mean}
\end{lemma}

\begin{proof}
Let $D_1$ be the standard gaussian distribution $\cN(0,1)$. 
Let $D_2$ be the distribution that with probability $1-\epsilon$ outputs a sample from $\cN(0,1)$ and with probability $\epsilon$, outputs $\sqrt{k} \epsilon^{-1/k}$. Observe that by construction, $d_{TV}(D_1, D_2) \leq \epsilon$.

The mean of $D_2$ is easily seen to be $\sqrt{k} \epsilon^{1-1/k}$. 
The variance of $D_2$ can be computed to be between $1$ and $(1 + k \epsilon^{1-2/k})$.

It is standard to check that $D_1$ is $C$-subgaussian in all moments for $C = 1 < 2$.
For $D_2$, we observe that $\E_{D_2} (x-\E_{D_2}x)^{2k}\leq \E_{D_2} x^{2k} \leq (1-\epsilon) k^{k/2} + \epsilon \sqrt{k}^{k} \epsilon^{-1} < 2 k^{k/2} \leq 2 k^{k/2} (\E_{D_2}(x-\E_{D_2}x)^2)^{k/2}$. Thus, $D_2$ is $2$-subgaussian.

This completes the proof. 

\end{proof}

Next, we show that our robust estimation error for the covariance and higher moments are also tight up to constant factors. As before, we will construct two zero-mean, $2$-subgaussian distributions on $\R$ that are at most $\epsilon$ apart in total variation but the variances differ by at least $k/2 \epsilon^{1-2/k}$. The construction is very similar to the one above for mean estimation.

\begin{lemma}
There exist two mean zero distributions $D_1$ and $D_2$ on $\R$ such that:
\begin{enumerate}
\item $D_1$, $D_2$ are both $2$-subgaussian in all moments up to $k$. 
\item $d_{TV}(D_1, D_2) \le \epsilon$.
\item For any $\epsilon < 1/2$, $|\E_{D_1}x^2 - \E_{D_2} x^2| \geq \frac{k}{2} \epsilon^{1-2/k}$.
\item For any $\epsilon < 2^{-k/2r}$, $|\E_{D_1}x^{2r} - \E_{D_2} x^{2r}| \geq \frac{k^{r}}{2} \epsilon^{1-2r/k}$.

\end{enumerate}
\label{lem:info-theoretic-lowerbound-covariance}
\end{lemma}

\begin{proof}
Let $D_1$ be the standard gaussian distribution $\cN(0,1)$. 
Let $D_2$ be the distribution that with probability $1-\epsilon$ outputs a sample from $\cN(0,1)$ and with probability $\epsilon$, outputs a uniform element of $\pm \sqrt{k} \epsilon^{-1/k}$.
Observe that by construction, $d_{TV}(D_1, D_2) \leq \epsilon$.

The mean of $D_1,D_2$ is easily seen to be $0$. The same computation as in the proof of Lemma  \ref{lem:info-theoretic-lowerbound-mean} establishes that $D_1$, $D_2$ are both $2$-subgaussian. 

The variance of $D_1$ is $1$ and for $D_2$ can be computed to be at least $(1-\epsilon)(1 + k \epsilon^{1-2/k} - k \epsilon^{2-2/k}) \geq 1+\frac{k}{2} \epsilon^{1-2/k}$.

The $2r$th moment of $D_1$ is $\sim k^{r}$ and for $D_2$ can be computed to as: $k^{r}(1-\epsilon + \epsilon^{1-2r/k}) \geq k^{r}(1+ 0.5 \epsilon^{1-2r/k})$ for $\epsilon < 2^{-k/2r}$.

This completes the proof. 

\end{proof}

%% file: content/basic-sos-proofs.tex
\section{Sum-of-squares toolkit}
\label{sec:basic-sos-proofs}

We will rely on the following sum-of-squares version of the inequality between the arithmetic and geometric mean.

\begin{lemma}[{\cite[Lemma A.1]{DBLP:conf/stoc/BarakKS15}}]
\label[lemma]{lem:AM-GM-sos}
$$
\Set{w_1 \geq 0,\ldots,w_n\ge 0}
\sststile{k}{w}
\Set{\prod_i w_i \leq \frac{\sum_i w_i^k}{k}}\,.
$$
\end{lemma}

\begin{lemma}
  \label[lemma]{lem:binomial}
  For all even $k$,
  \begin{equation}
    \sststile{k}{a,b} \Set{(a+b)^k\le 2^{k-1}(a^k + b^k)}
    \,.
  \end{equation}
\end{lemma}

\begin{proof}
  The lemma holds for $k=2$, because
  \begin{equation}
    \sststile{2}{a,b} \Set{2a^2+2b^2 - (a+b)^2 = a^2 + b^2 - 2ab = (a-b)^2 \ge 0}
    \,.
  \end{equation}
  By iterating this inequality $k/2$ times, we obtain
  \begin{equation}
    \sststile{k}{a,b} \Set{(a+b)^k = ((a+b)^2)^{k/2} \le 2^{k/2}(a^2 + b^2)^{k/2}}\,.
  \end{equation}
  By the binomial identity and \cref{lem:AM-GM-sos},
  \begin{equation}
    \sststile{k}{a,b} \quad
    \begin{aligned}[t]
      (a^2 + b^2)^{k/2}
      &= \sum_{i=0}^{k/2}\binom{k}{i} a^{2i}b^{k-2i}\\
      &\le \sum_{i=0}^{k/2}\binom{k}{i} \Paren{\tfrac {2i}{k} \cdot a^{k} + \tfrac{k-2i}{k} \cdot b^k}\\
      &= \sum_{i=0}^{k/2}\binom{k}{i} \Paren{\tfrac {1}{2} \cdot a^{k} + \tfrac 1{2} \cdot b^k}\\
      &= 2^{k/2-1}\cdot (a^k + b^k) \,.
    \end{aligned}
  \end{equation}
  Here, we used for the second equality that $\binom{k}{i}=\binom{k}{k-i}$.
  
  Combining the previous two inequalities yields the desired inequality,
  \begin{equation}
    \sststile{k}{a,b}
    (a+b)^k \le 2^{k/2} \cdot (a^2 + b^2)^{k/2} \le 2^{k-1}\cdot (a^k+b^k)\,.
  \end{equation}
\end{proof}

\begin{lemma}
  \label[lemma]{lem:sos-fact-2}
  For any even $k$,
  \begin{equation}
    \Set{f^k \leq 1}
    \sststile{k}{f}
    \Set{f \leq 1}\,.
  \end{equation}
\end{lemma}

\begin{proof}
  We apply \cref{lem:AM-GM-sos} for $k/2$ and substitute $w_1 = f^2$ and $w_2, w_3, \ldots, w_{k/2} = 1$ to obtain $\sststile{k}{f} f^2 \leq \frac{f^k+(k/2-1)}{k/2}$.
  Thus, 
  \begin{equation}
    \Set{f^k \leq 1} \sststile{k}{f} \Set{f^2 \leq 1}. \label{eq:reduce-deg-to-2}
  \end{equation}
  At the same time, since $2-2f = 1-f^2 + (1-f)^2$, we have 
  \begin{equation}
    \Set{f^2 \leq 1} \sststile{2}{f} \Set{1-f \geq 0}. \label{eq:reduce-deg-to-1}
  \end{equation}
   Combining \eqref{eq:reduce-deg-to-2} and \eqref{eq:reduce-deg-to-1} yields the desired conclusion.
\end{proof}

\begin{lemma} 
  \label[lemma]{lem:sos-fact-1}
  For every even $k\in \N$ and all $\delta<0.1$,
  \begin{equation} 
    \Set{(f-1)^k \leq \delta^k (f+1)^k} \sststile{k}{f} \Set{1-\delta' \le f \leq 1+\delta'},
  \end{equation}
  where $\delta' = 100 \delta$.
\end{lemma}
\begin{proof}
  By \cref{lem:binomial},
  \begin{displaymath}
    \sststile{k}{f} (f+1)^k = ((f-1) + 2)^k \le 2^{k} (f-1)^k + 4^{k}\,.
  \end{displaymath}
  Therefore, for $\delta'=\Paren{\frac {4\delta}{1-2\delta}}^k$,
  \begin{displaymath}
    \Set{(f-1)^k \leq \delta^k (f+1)^k} ~\sststile{k}{f}~ \Set{(f-1)^k \le (\delta')^k}\,.
  \end{displaymath}
  By \cref{lem:sos-fact-2}, $\Set{(f-1)^k \le (\delta')^k}\sststile{}{} \Set{ 1-\delta' \le f \le 1+\delta}$, which together with the previous inequality gives the desired inequality.
\end{proof}

\begin{proposition}
  Let $\cA$ be a set of polynomial equality axioms in variable $x$ such that: 
  \[
    \cA \sststile{2t}{} p(x,u) \geq 0,
  \]
  for a polynomial $p$ with total degree at most $2t$.

  Then, for any pseudo-distribution $D$ of degree $2t$ on $x$ satisfying $\cA$,
  $\sststile{2t}{u} \pE_{D} p(x,u) \geq 0$.\label[proposition]{prop:SoS-proof-for-unbounded-variable}
\end{proposition}

\begin{proof}
  Let $\cA = \{q_1 = 0, q_2 = 0, \ldots, q_r = 0\}$. 
  Since $\cA \sststile{2t}{} p(x,u) \geq 0$, there are degree $\leq 2t$ (in $x$) polynomials $g_1(x,u), g_2(x,u), \ldots g_{r'}(x,u)$ and polynomials $h_1(x,u), \ldots, h_{r}(x,u)$ such that 
  $p(x,u) = \sum_i g_i(x,u)^2 + \sum_{j} h_j(x,u) q_j(x)$ where the degree of any $h_j(x,u) \cdot q_j(x)$ in $x$ is at most $2t$.

  Taking pseudo-expectations with respect to $D$, we have: 
  $\pE_D p(x,u) = \sum_i \pE_D g_i(x,u)^2 + \sum_j \pE_D h_j(x,u) q_j(x) = \sum_i \pE_D g_i(x,u)^2,$ 
  where we used that the second term must vanish as $D$ satisfies $\cA$.

  Now, observe that by matching degrees on both sides as polynomial in $u$, no $\pE_D g_i(x,u)^2$ has degree more than $2t$. 

  Further, each $g_i(x,u)^2$ can be written as $\iprod{G(x)G(x)^{\top}, (1,u)^{\otimes 2t}}$ for some matrix valued polynomial $G(x)$.

  Taking pseudo-expectations with respect to $D$ yields that: 
  $\pE_D p(x,u) = \sum_i \iprod{\pE_D[G(x)G(x)^{\top}], (1,u)^{\otimes 2t}}$. 

  Now, we claim that $\pE_D G(x) G(x)^{\top}$ is psd. To see this, let $w$ be any vector and consider the quadratic form of this matrix in $w$. Then, the quadratic form equals $\pE_D \iprod{G(x), w}^2$ which is non-negative by the positivity of pseudo-expectations. 

  Thus, $\pE_D p(x,u)$ is a sum of squares polynomial in $u$ implying that $\sststile{2t}{u} \pE_D p(x,u) \geq 0$.
\end{proof}